\documentclass{elsarticle}
\usepackage{amsmath,amssymb,a4wide,psfrag}
\newcommand{\Tr}[0]{\text{Tr}\;}
\newcommand{\sign}[0]{\text{sign}\;}

\newcommand{\ket}[1]{\left|#1 \right\rangle}
\newcommand{\bra}[1]{\left\langle#1 \right|}
\newcommand{\braket}[2]{\left\langle #1|#2 \right\rangle}

\newtheorem{theorem}{Theorem}[section]
\newtheorem{lemma}[theorem]{Lemma}

\newenvironment{proof}[1][Proof]{\begin{trivlist}
\item[\hskip \labelsep {\bfseries #1}]}{\end{trivlist}}

\renewcommand{\qed}{\nobreak \ifvmode \relax \else
      \ifdim\lastskip<1.5em \hskip-\lastskip
      \hskip1.5em plus0em minus0.5em \fi \nobreak
      \vrule height0.75em width0.5em depth0.25em\fi}

\begin{document}
\begin{frontmatter}
\title{A Ginsparg-Wilson approach to lattice CP symmetry, Weyl and Majoranna fermions, and the Higgs mechanism}
\author[a,b]{Nigel Cundy}
\address[a]{Institut f\"ur Theoretische Physik, Universit\"at Regensburg, D-93040 Regensburg, Germany}
\address[b]{Lattice Gauge Theory Research Center, FPRD, and CTP,\\ Department of Physics \&
    Astronomy, Seoul National University, Seoul, 151-747,\\ South Korea}
\date{\today}
\begin{abstract}
Recently, two solutions have been proposed to the long standing problem of $\mathcal{CP}$-symmetry on the lattice, which is particularly evident when considering the construction of chiral gauge theories. The first, based on a lattice modification of $\mathcal{CP}$ was presented by Igarashi and Pawlowski; the second by myself using the renormalisation group and Ginsparg-Wilson relation. In this work, I combine the two approaches and show that they are each part of a more general framework related to an underlying renormalisation group. I continue by formulating Weyl and Majorana fermions on the lattice, and discussing applications to the fermion propagator in the presence of the Higgs field. This resolves various difficulties when the standard continuum $\mathcal{CP}$ is applied to the lattice standard model, in particular concerning a non-local shift in the quark propagator. The modified lattice $\mathcal{CP}$ resolves the difficulties where standard $\mathcal{CP}$ broke down when applied to a lattice theory. However, although all actions and observables are invariant under $\mathcal{CP}$, this formulation gives a non-local generator of lattice $\mathcal{CP}$.    
\end{abstract}
\begin{keyword}
Chiral fermions \sep Lattice QCD  \sep Renormalisation group
\PACS  11.30.Rd \sep 11.15.Ha \sep 11.10.Hi 
\end{keyword}
\end{frontmatter}

\section{Introduction}
The procedure for constructing a lattice equivalent of a continuum symmetry was discovered by Ginsparg and Wilson in 1982~\cite{Ginsparg:1982bj}. Its best known application is to chiral symmetry~\cite{Hasenfratz:1998jp}, but in principle the same procedure can be applied to any infinitesimal continuum symmetry. It is not obvious, however, that it can be directly applied to discrete symmetries, such as the charge conjugation, parity and time reversal symmetries. Until recently, it has been assumed that these symmetries carried directly from the continuum to the lattice in the same form. If we interpret lattice QCD as a well defined quantum field theory which gives equivalent physics to continuum QCD in a particular limit of its parameters, then there is no obvious reason to suggest modifying $\mathcal{CP}$ symmetry on the lattice. There is, however, an alternative interpretation of lattice QCD, implicit in the Ginsparg-Wilson relation and the construction of fixed point fermions, that it is a quantum field theory with a particular regulator obtained by blocking in some renormalisation group scheme from an equivalent theory to continuum QCD, and in this interpretation, if the blocking is not $\mathcal{CP}$ invariant, it is clear that $\mathcal{CP}$ symmetry should be modified on the lattice. Recently, blockings to construct an equivalent of the fermionic part of overlap lattice QCD have been derived~\cite{Cundy:2009ab}. In this work, I use these blockings to construct an appropriate $\mathcal{CP}$ symmetry on the lattice.   

The Ginsparg-Wilson procedure uses a blocked renormalisation group transformation~\cite{Wilson:1977} to convert the fermion fields from one scale to another, for example from the continuum to the lattice. The operator
\begin{gather}
\int d\psi_1 d\overline{\psi}_1 e^{-(\overline{\psi}_1 - \overline{\psi}_0 B^{-1}) \alpha (\psi_1 - B^{-1} \psi_0)}\nonumber
\end{gather} 
is introduced into the generating function (which is constructed in terms of the original fermion fields $\psi_0$), and the integration is performed over the old fermion fields to give a new generating function in terms of the new fermion fields $\psi_1$. This procedure was used to construct the chirally symmetric fixed-point action~\cite{Hasenfratz:1994, Bietenholz:1995cy}. Its application to other lattice chiral Dirac operators, such as the overlap operator, was not realised until a recent study~\cite{Cundy:2009ab} which found blocking matrices $\overline{B}$, $B$ and $\alpha$ which could be used to construct the overlap Dirac operator (and could be extended to other chiral Dirac operators based on the overlap formalism~\cite{Cundy:2008cs,Cundy:2008cn}) from the continuum (although this blocking only transforms the fermion fields while leaving the partition function itself unchanged once the fermion fields have been integrated out, and so far there is no complete renormalisation group construction of an overlap lattice action from the continuum). It now seems that all lattice chiral fermions are closely related to the renormalisation group. 

The history of chiral fermions on the lattice, of course, suggests otherwise. The initial effort to cheat the Neilson-Ninoyoma no-go theorem~\cite{Nishy-Ninny}, which states that it is impossible to have a difference between the number of right and left handed fermion fields on the lattice if the Dirac operator is local, translation invariant and has the correct continuum limit, was by Kaplan with his  domain wall fermions~\cite{Kaplan:1992bt,Shamir:1993zy,Furman:1994ky}; where the lattice was extended to five dimensions, and the unwanted left handed fermions separated from the four dimensional surface used for QCD by a large fifth dimension; once the size of the fifth dimension reaches infinity chiral symmetry is restored, with finite fifth dimension it is only approximate. Inspired by this work, Neuberger and Narayanan constructed a chiral Dirac operator as the overlap between two vacua, which led them to the overlap formula~\cite{Narayanan:1993sk,Narayanan:1993ss,Neuberger:1998my,Neuberger:1998fp},
\begin{gather}
D = 1+\mu + (1-\mu)\gamma_5\sign(K),
\end{gather}
where $\mu = m/(2m_W)$ is a mass parameter, and $K[m_W]$ is the Hermitian form of another lattice Dirac operator, usually the Wilson operator. The domain wall Dirac operator reduces to a form of the overlap operator in infinite fifth dimension, and can thus be thought of as a representation of one particular inexact approximation of the overlap. Many other lattice Dirac operators with exact chiral symmetry have been found~\cite{Cundy:2008cs,Kerler:2002xk,Kerler:2003qv,Fujikawa:2000my,Fujikawa:2001fb}, but the overlap operator remains the only one in practical use.

After the overlap operator was found to satisfy the simplest form of the Ginsparg-Wilson relation, Martin L\"uscher constructed the lattice chiral symmetry transformations~\cite{Luscher:1998pqa} for the massless Dirac operator. However, it is important to emphasize, particularly in the context of this work, that the canonical Ginsparg-Wilson relation
\begin{gather}
\gamma_5 D + D \gamma_5 (1-aD) = 0
\end{gather}
and the corresponding chiral symmetry transformations is only one of an infinite number of ways to describe chiral symmetry on the lattice.

One troubling problem has remained for the Ginsparg-Wilson construction: $\mathcal{CP}$ symmetry. This is perhaps most clearly seen within the context of a chiral gauge theory, where the action
\begin{gather}
\overline{\psi} D (1+\gamma_5) \psi\nonumber
\end{gather} 
transforms under standard $\mathcal{CP}$ to
\begin{gather}
\overline{\psi} D (1+\gamma_5(1-aD)) \psi\nonumber,
\end{gather} 
which, unlike the continuum chiral gauge theory, is clearly not invariant at non-zero lattice spacing. This is related to the observation that the fermion and anti-fermion lattice chiral symmetry transformations are not symmetric and do not respect the standard $\mathcal{CP}$ symmetry.

Two no-go theorems, similar to the Neilson-Ninoyama theorem, have been constructed to say that it is impossible to construct a $\mathcal{CP}$-invariant lattice chiral gauge theory under certain reasonable conditions~\cite{Fujikawa:2002is,Jahn:2002kg}. It is tempting to just dismiss this as an interesting and inconvenient anomaly, but there are more serious concerns. A detailed study of the effects of $\mathcal{CP}$ breaking on the lattice was carried out by Fujikawa and his collaborators~\cite{Fujikawa:2002vj}, and they found three effects of $\mathcal{CP}$ breaking: (I) a constant phase in the generating functional; (II) a constant multiplied to the generating functional; and (III) a shift in the fermion propagator. The first two are unimportant, but the third is concerning. It is tempting to treat $\mathcal{CP}$ invariance in the same way that the lattice treats translation invariance, i.e. restored in the continuum limit, but there is one problem with this approach, namely that in the presence of a Higgs field, for example in the electro-weak theory, the shift in the propagator is non-local, which usually implies that there is no smooth continuum limit. A lack of $\mathcal{CP}$ invariance also creates problems constructing Weyl and Majoranna actions.

The question of $\mathcal{CP}$ on the lattice has recently been revived by three different studies, each giving a different solution. The first was by Gattringer and Pak~\cite{Gattringer:2007dz,Gattringer:2008je} (inspired by an idea in~\cite{Hasenfratz:2007dp}), who increased the number of fermionic degrees of freedom. I shall not consider this work here. The second solution was by Igarashi and Pawlowski~\cite{Igarashi:2009kj,Igarashi:2009wr}, who realised that a remnant of $\mathcal{CP}$ symmetry could be restored on the lattice in the same way that chiral symmetry is restored on the lattice by modifying the $\mathcal{CP}$ transformations. Although this approach works well, the motivations for doing this, other than to make $\mathcal{CP}$ work, were not well explored: the different $\mathcal{CP}$ transformation was applied arbitrarily rather than as a consequence of some more general approach, and the modified $\mathcal{CP}$ was applied to the left and right handed fermion fields rather than the original fermion fields and thus cannot be applied to the chiral symmetry transformation itself. The third work was my own~\cite{Cundy:2009ab}, where, as a footnote to a study of the connection between the renormalisation group and overlap fermions, I constructed a $\mathcal{CP}$ invariant chiral gauge theory using the standard continuum $\mathcal{CP}$ operators and a non-standard form of the Ginsparg Wilson relation. This solution followed directly from the Ginsparg-Wilson procedure and the renormalisation group. I showed that the measure of the chiral gauge theory was gauge invariant, necessary in the construction, but I did not directly demonstrate that my chiral projection operators were local for the $\mathcal{CP}$ invariant theory, only for the standard construction of lattice chiral symmetry. My chiral gauge theory was only valid for this one particular form of the Ginsparg-Wilson equation.\footnote{After this work was prepared, another work on this subject was presented, using mirror fermions~\cite{Poppitz:2010at}.}

In this work, I shall turn to the question of $\mathcal{CP}$ on the lattice more directly and in more detail. The basis of this work is the renormalisation group construction of the overlap operator described in my earlier paper, although the final approach is similar to that suggested by Igarashi and Pawlowski. I shall show that the Ginsparg-Wilson relation itself requires that $\mathcal{CP}$ symmetry should be modified on the lattice, and the form of lattice $\mathcal{CP}$ for a particular lattice chiral symmetry flows directly from the Ginsparg-Wilson relation. I will explicitly construct the required modification for an entire group of lattice chiral symmetries. I shall show that all lattice observables do, in fact, satisfy $\mathcal{CP}$ symmetry, but only if the correct lattice $\mathcal{CP}$ symmetry is used. I shall consider the anomalies discovered in~\cite{Fujikawa:2002vj}, and show that they are all resolved with this modified $\mathcal{CP}$ symmetry. Finally, I shall construct a chiral gauge theory, a Weyl action, a Majoranna action on the lattice, and consider the Yukawa coupling between the lattice Higgs field and the fermions within the context of the electroweak theory. However, this construction is only valid for Ginsparg-Wilson fermions, and for other lattice actions the $\mathcal{CP}$ anomalies will remain.

In section \ref{sec:2}, I overview the Ginsparg-Wilson relation and the renormalisation group construction of the lattice overlap action. In section \ref{sec:2a}, I construct the lattice $\mathcal{CP}$ operators, and in section \ref{sec:2b} consider the effects of $\mathcal{CP}$ on the generating functional. In section \ref{sec:3}, I consider the locality of my chiral symmetry transformations, and in section \ref{sec:4} I digress into a discussion of the conserved currents and Ward identities associated with the various lattice chiral symmetries. In sections  \ref{sec:7} and \ref{sec:8} I construct Weyl and Majoranna actions, and consider fermion propagators in the presence of a Higgs field in section \ref{sec:8b}. After my conclusions, there are three appendices establishing notation, and describing the Ginsparg-Wilson operators.

\section{The generalised Ginsparg-Wilson relation}\label{sec:2}
$D$ is the massless overlap operator~\cite{Neuberger:1998my},
\begin{gather}
D = 1+\gamma_5\sign(K),\label{eq:ov}
\end{gather}
whose eigenvalues lie on the unit circle of radius 1 centred at 1, and $K$ is some valid kernel operator. One commonly used kernel in lattice gauge theory simulations is the lattice Wilson Dirac operator,
\begin{gather}
K_{xy} = \gamma_5[\delta_{xy} - \kappa\sum_{\mu}(1-\gamma_{\mu} )U_{\mu}(x) \delta_{x+\mu,y} + (1+\gamma_{\mu})U^{\dagger}_{\mu}(x-\mu)\delta_{x-\mu,y}],
\end{gather}
with the Wilson mass, $m_W = 4-1/(2\kappa)$, satisfying $0<m_W<2$. In this work, however, I require an overlap operator which is well defined in the continuum. This means that the kernel should be taken from a blocked continuum Dirac operator (plus an additional term to account for the non-lattice degrees of freedom) which is either equivalent to or smoothly reduces to a lattice Dirac operator in a particular limit (here I use an operator which reduces to the Wilson Dirac operator). Defining the lattice Dirac operator as a limit of a continuum operator means that it is possible for the blockings to be invertible, and that both the continuum and `lattice' Dirac operators have the same number of degrees of freedom, including the same number of eigenvalues.  One possible example is given in~\cite{Cundy:2009ab}, although for this work the precise details of the kernel operator are unimportant. The kernel operator itself does not need to be analytic in the lattice limit as long as the overlap operator itself is well defined and equivalent to the target lattice theory. I shall refer to this Dirac operator $D$ as the lattice Dirac operator, even though, strictly, I am always using continuum fermion fields. 

Ginsparg and Wilson originally derived their famous equation using Wilson's formulation of the renormalisation group~\cite{Ginsparg:1982bj}. They considered how a symmetry is affected by a renormalisation group blocking. This formulation is more general than just the application to chiral symmetry in lattice QCD, and it includes the various blockings which can be used to convert continuum QCD to lattice QCD. The process considers a renormalisation group of the partition function  such as 
\begin{align}
Z(J_0,\overline{J}_0,U,g_1) =& \int d\psi_0 d\overline{\psi}_0 e^{-\frac{1}{4g_0^2} F_{\mu\nu}^2 - \overline{\psi}_0 D_0 \psi_0 + \overline{J}_0 \psi_0 + \overline{\psi}_0 J_0}\int d\psi_1 d\overline{\psi}_1 e^{-(\overline{\psi}_0 -  \overline{\psi}_1\overline{B}^{-1})\alpha(\psi_0 - B^{-1} \psi_1)}\\
=&\int d\psi_1 d \overline{\psi}_1 e^{-\frac{1}{4g_1^2} F_{\mu\nu}^2  - \overline{\psi_1} D_1 \psi_1 + \overline{\psi}_1 \overline{B}'J_0 + \overline{J}_0 B' \psi_1}
,
\intertext{where the original generating functional was,}
Z(J_0,\overline{J}_0,U,g_0) =&\int d\psi_0 d\overline{\psi}_0 e^{\frac{1}{-4g_0^2} F_{\mu\nu}^2 - \overline{\psi}_0 D_0 \psi_0 + \overline{J}_0 \psi_0 + \overline{\psi}_0 J_0},
\end{align}
and $D_1$, $B'$ and $\overline{B}'$ can be calculated by integrating out the $\psi_0$ and $\overline{\psi}_0$ fields. As a simplifying (though not necessary) assumption, I shall assume that $B^{-1}$ and $\overline{B}^{-1}$ are invertible, and use $B' = B$, $\overline{B}' = \overline{B}$. I assume that the Jacobian created from this integration can be absorbed into a modification of the gauge coupling.
Demanding that the original action is invariant under chiral symmetry leads to a Ginsparg-Wilson relation~\cite{Cundy:2009ab,Cundy:2008cn,Borici:2007ft}
\begin{gather}
\overline{B}\gamma_5\overline{B}^{-1} D + D B^{-1}\gamma_5B = D(B^{-1}\gamma_5B\alpha^{-1} + \alpha^{-1}\overline{B}\gamma_5\overline{B}^{-1}) D,\label{eq:2}
\end{gather}
for invertible local operators $B$ and $\overline{B}$ which do not necessarily have to commute with $\gamma_5$ (Ginsparg and Wilson's original formulation, and its use in constructing the perfect action~\cite{Hasenfratz:1994}, assumed that $B$ and $\overline{B}$ commute with $\gamma_5$). In~\cite{Cundy:2009ab}, I showed that a continuum form of the overlap operator which smoothly reduced to the lattice operator in a certain limit could be derived using such an approach. In this operator, the Dirac operator was constructed so that there was a decoupling between eigenvectors entirely on the lattice sites and eigenvectors entirely off-lattice. The eigenvalues of the off-lattice part of the Dirac operator are then projected to $2/a$, so that they do not contribute to the fermion propagators. In this way, the the theory was constructed without reducing the number of degrees of freedom but in such a way that the dynamics of the theory could be treated on the lattice while the off lattice site correction computed analytically (and trivially, since the determinant is just a constant). This is analogous to Wilson's approach to remove the fermion doublers. This overlap operator satisfies the Ginsparg-Wilson relation for various possible functions $\alpha$ and $B$.  This relationship between the renormalisation group and the overlap operator explains why the Ginsparg-Wilson equation can be applied to overlap lattice QCD. If certain lattice operators were not connected to the continuum via the renormalisation group and the Ginsparg-Wilson equation only fulfilled coincidentally, then much of the physical significance of the relation is lost. Equation (\ref{eq:2}) can be re-expressed as
\begin{gather}
\gamma_L D + D \gamma_R = 0,
\end{gather}
where, for example,
\begin{align}
\gamma_L =& \overline{B}\gamma_5\overline{B}^{-1}\nonumber\\
\gamma_R = & B^{-1}\gamma_5B - (B^{-1}\gamma_5B\alpha^{-1} + \alpha^{-1}\overline{B}\gamma_5\overline{B}^{-1}) D.
\end{align}

There are, in fact, an infinite number of possible solutions $\gamma_L$ and $\gamma_R$ for a given chiral lattice Dirac operator, which only become equivalent in the continuum limit~\cite{Mandula:2007jt,Mandula:2009yd}. Following the formulation of~\cite{Cundy:2009ab}, I here set $\alpha = \infty$ so that for finite and non-zero $B$ and $\overline{B}$ the right hand side of equation (\ref{eq:2}) vanishes, and the operators $\gamma_L$ and $\gamma_R$ can be easily constructed from $B$, $\overline{B}$ and $\gamma_5$. In principle, any of these chiral symmetries could be used, each with their own conserved current to take. Which of these chiral symmetries is `correct'? In a sense, they all are, since they all give the correct continuum limit and the ambiguities disappear in that limit. But the additional symmetries may yet cause difficulties at finite lattice spacing. The problem is that the generators of these symmetries are not independent, so in principle at non-zero lattice spacing the different currents will mix and one has to consider the entire infinite group of chiral symmetries. The reason for this infinite group is that $\gamma_L$ and $\gamma_R$ are not restricted to conjugate operators in Euclidean space, unlike in Minkowski space. This means that under standard $\mathcal{CP}$ symmetry, one representation of the chiral symmetry will transform into a different representation, and these two representations need not give the same physics at finite lattice spacing, as discussed in~\cite{Mandula:2007jt,Mandula:2009yd}. This is the root of the problem of $\mathcal{CP}$ symmetry on the lattice. These differences will disappear in the continuum limit. It is also strongly suggested that the differences will not be present in the renormalised theory, because the different chiral symmetries arise from different representations of the renormalisation group derivation of the overlap operator~\cite{Cundy:2009ab}. However, for the lattice $\mathcal{CP}$ problem to be removed, either some $\mathcal{CP}$-symmetric formulation of lattice chiral symmetry much be constructed or $\mathcal{CP}$-symmetry treated using an analogue to the Ginsparg-Wilson procedure. Since the best lattice theory respects, as much as possible, the symmetries of the continuum theory, the preferred chiral symmetry would be one that does respect $\mathcal{CP}$-symmetry, if such a symmetry exits. 

It is tempting to suggest that since the 'real world' (Minkowski) has conjugate fermion variables, the best chiral symmetry should be one where the Euclidean variables are transformed in some symmetric way. Two solutions consistent with conjugate fermion fields were considered in~\cite{Mandula:2007jt,Mandula:2009yd}, $\gamma_L = \gamma_R = \gamma_5(1-D/2)$ and $\gamma_L = (1-D)^{1/2}\gamma_5; \gamma_R= \gamma_5 (1-D)^{1/2}$; however the first of these is not invertible and the second is ill-defined. A third group of chiral symmetry transformations was described in~\cite{Fujikawa:2002vj}, where $\gamma'_L = (1-sD)\gamma_5$, $\gamma'_R = \gamma_5 (1-(1-s)D$, and $\gamma_L = \gamma'_L/\sqrt{(\gamma'_L)^2}$, $\gamma_R = \gamma'_R/\sqrt{(\gamma'_R)^2}$. It was noted in that work that the symmetric version, at $s = 1/2$, contains potential difficulties concerning locality due to the zero mode doublers (with eigenvalue $\lambda = 2/a$) of the Dirac operator. Thus finding a lattice chiral symmetry with symmetric $\gamma$ operators seems problematic; and two no-go theorems have been constructed to demonstrate its impossibility~\cite{Fujikawa:2002is,Jahn:2002kg}.\footnote{However, both these no go theorems seem to assume that the chiral projectors are continuous functions of the parameter used to interpolate between $\gamma_L$ and $\gamma_R$. For example, if we write that $\gamma^{(\upsilon)}_R = \sign((1+\upsilon)\gamma_5 - (1-\upsilon)\sign(K))$, and $\gamma_L^{(\upsilon)} = \sign((1-\upsilon)\gamma_5 - (1+\upsilon)\gamma_5\sign(K)\gamma_5)$, then these operators are both discontinuous at $\upsilon = 0$ in the presence of zero modes (this can be shown by considering the zero mode doublers). \cite{Fujikawa:2002is} shows that if $\gamma_R^{(\upsilon)}\gamma_5 = \gamma_5\gamma_L^{(\upsilon)}$, then the projection operator is non-local. $\mathcal{CP}$ symmetry converts $\gamma_R^{(\upsilon)} \rightarrow - \gamma_L^{(-\upsilon)}$ and the $\mathcal{CP}$ breaking term is of order $\upsilon$. The $\mathcal{CP}$ symmetric case at $\upsilon = 0$ seems to be ruled out. However, if the the operators are discontinuous at $\upsilon = 0$ then the $\mathcal{CP}$ breaking can be removed by taking a limit towards $\upsilon \rightarrow 0$ while always maintaining a difference between the two operators. The more sophisticated analysis of \cite{Jahn:2002kg} makes the same assumption. It demonstrates that the two projectors must be in two topologically distinct states, and assumes that the only winding comes from the discontinuity in the Dirac operator. If however, the projectors are also discontinuous, then there are additional discontinuities in the function $g$ used to to interpolate between the projector $1\pm \gamma_R$ and the projector $1\mp \gamma_L$. This means that the difference of the Chern index of the two projectors need not be the topological charge associated with the Dirac operator. However, I have not found any local operators which sidestep the no-go theorem by these means, and my approach in this work takes an entirely different direction.}

In~\cite{Igarashi:2009kj}, a different approach to the lattice $\mathcal{CP}$ problem was proposed: to modify the way that lattice fermion fields transform under $\mathcal{CP}$. However, although the approach works, it is unclear from that work why this approach should be taken. Ideally, the solution to lattice $\mathcal{CP}$ should flow naturally from the Ginsparg-Wilson relation. The solution of~\cite{Igarashi:2009kj} also only discussed the standard Ginsparg-Wilson relation $\gamma_L = \gamma_5; \gamma_R = \gamma_5 (1-D)$ without any reference to the larger group of chiral symmetry transformations. I shall later re-derive and extend their result based on a different approach inspired by the methods of Ginsparg and Wilson.

In~\cite{Cundy:2009ab}, my considerations of the renormalisation group led me towards a `natural' group of lattice chiral symmetry transformations. The symmetry transformations were expressed in terms of a continuum form of the overlap operator and the usual continuum Dirac operator, $D_0$, $\gamma_R = D^{-(1+\eta)/2} Z D_0^{(1+\eta)/2}\gamma_5D_0^{-(1+\eta)/2} Z^{\dagger}D^{(1+\eta)/2}$, where $\eta$ is an arbitrary tunable parameter and the unitary operator $Z$ projects the eigenvectors of $D_0$ onto the eigenvectors of $D$. Several possible choices of $Z$ are given in~\cite{Cundy:2009ab}; the simplest (though not practical or obviously local) is $Z=\ket{g_i^+}\bra{\tilde{g}_i^+} + \ket{g_i^-}\bra{\tilde{g}_i^-}$, where $\ket{g_i^{\pm}}$ are the chiral eigenvectors of $D^{\dagger}D$ and $\ket{\tilde{g}_i^{\pm}}$ are (in some sense corresponding) chiral eigenvectors of $D_0^{\dagger}D_0$. This choice of $Z$ obviously commutes with $\gamma_5$ ($[Z,\gamma_5]=0$ ensures that $D = \overline{B}^{(\eta)}D_0B^{(\eta)} = \overline{B}^{(-\eta)}D_0B^{(-\eta)}$ is $\gamma_5$-Hermitian). The blockings used to generate this $\gamma_R$ are $B^{(\eta)}=D_0^{-(1+\eta)/2}Z^{\dagger}D^{(1+\eta)/2}$ and $\overline{B}^{(\eta)} = D^{(1-\eta)/2}ZD_0^{-(1-\eta)/2}$ with $\alpha = \infty$. Using the matrix representation of the eigenvalue decomposition of $D$, as discussed in appendix \ref{app:eigenvalues}, and using the simplification that the matrix components of $D_0$ are real, then, as outlined in appendix \ref{app:eigenvalues} (equations \ref{eq:Cga} - \ref{eq:Cgb}), the class of $\gamma$ matrices can be expressed as
\begin{align}
{\gamma}^{(\eta)}_R = & \gamma_5 \cos\left[\frac{1}{2}(1+\eta)(\pi - 2 \theta)\right] + \sign(\gamma_5 (D^{\dagger}-D)) \sin\left[ \frac{1}{2}(1+\eta)(\pi - 2 \theta)\right]\nonumber\\
{\gamma}^{(\eta)}_L = & \gamma_5 \cos\left[\frac{1}{2}(\eta-1)(\pi - 2 \theta)\right] + \sign(\gamma_5(D^{\dagger}-D)) \sin\left[ \frac{1}{2}(\eta-1)(\pi - 2 \theta)\right]\label{eq:8a}\\
\tan\theta =& 2 \frac{\sqrt{1-a^2D^{\dagger} D/4}}{\sqrt{a^2 D^{\dagger} D}}\label{eq:gamma_eta0}
\end{align}
Note that $D^{\dagger}D$ commutes with both $\gamma_5$ and $\sign(K)$. $\gamma_5 = \gamma_R^{(-1)} = \gamma_L^{(1)}$. $\gamma_L^{(\eta)} = \gamma_5 \gamma_R^{(-\eta)}\gamma_5 \forall\eta$. The standard form of lattice chiral symmetry is the solution at $\eta = \pm 1$, while the symmetric form is at $\eta = 0$:
\begin{align}
\gamma^{(0)}_L =& \gamma_5 \sqrt{1-a^2D^{\dagger}D/4} + \frac{a}{4}\gamma_5 (D-D^{\dagger})(1-a^2D^{\dagger}D/4)^{-1/2}\nonumber\\
\gamma^{(0)}_R =& \gamma_5 \sqrt{1-a^2D^{\dagger}D/4} - \frac{a}{4}\gamma_5 (D-D^{\dagger})(1-a^2D^{\dagger}D/4)^{-1/2}.\label{eq:glgr}
\end{align}
This symmetric form can also be written as $\gamma^{(0)}_R = \sign(\gamma_5 - \sign(K))$.  
It is immediately clear that there are potential problems with locality, since $(D-D^{\dagger})(1-a^2D^{\dagger}D/4)^{-1/2}$ is a doubler free Dirac operator which apparently anti-commutes with $\gamma_5$\footnote{The ambiguity in the definition of $\gamma_R^{(0)}$ for those eigenvectors of $D$ and $D^{\dagger}$ with eigenvalue two means that, depending on how these eigenvectors are treated, this operator may not, in fact, anti-commutate with $\gamma_5$. This is why the anti-commutation may only be apparent.}, and by the Nielson-Ninoyama theorem is therefore non-local. A discussion of the locality of these operators is thus essential, and will be given in section \ref{sec:3}. In equation (\ref{eq:gamma_eta0}) the difficulties occur when $D$ has an eigenvalue of 2, which corresponds to $\theta = 0$. For the zero modes of $D$, $\theta = \pi/2$ and the second term in equation (\ref{eq:8a}) is zero: there is no ambiguity in the definition of the operator in this case. Any problems with an ambiguous definition, and the related issue of locality, are avoided if $\eta$ takes an odd integer value, because the coefficient of $\sign(\gamma_5(D-D^{\dagger}))$ is zero for both the zero modes and the eigenvalues at $\lambda = 2/a$. 

At other $\eta$, these operators have one obvious problem: they are ill defined at eigenvalues of $D^{\dagger}D = 4$. These are the zero mode doublers, which are inevitable if we move outside the trivial topological sector. It is tempting to say that this is unimportant, given that there are solutions which are valid, but doing so leads to certain problems which we shall encounter later, in the definition of $\mathcal{CP}$-symmetry on the lattice. In practice, for non-odd integer values of $\eta$ one can create a well defined operator by deflating the zero mode doublers, and replacing them with $\gamma_5$, for example by writing (for non-integer $\eta+\epsilon$, where $\epsilon$ is some infinitesimal real number used to shift the ambiguity away from integer $\eta$)
\begin{align}
{\gamma_R'}^{\eta} =& \gamma_R^{\eta} (1-\ket{\psi_2}\bra{\psi_2}) - \sign(\sin({(\eta+\epsilon)\pi}/{2})) \gamma_5\ket{\psi_2}\bra{\psi_2}\nonumber\\
 {\gamma_L'}^{\eta} =& \gamma_L^{\eta} (1-\ket{\psi_2}\bra{\psi_2}) +\sign(\sin({(\eta+\epsilon)\pi}/{2})) \gamma_5\ket{\psi_2}\bra{\psi_2}.\label{eq:gamma_eta1}
\end{align}
This will not, of course, be local for all but certain specific values of $\eta$.

There is, however, an additional concern involving the zero mode doublers during the construction of these blockings. To write the $\gamma$ matrices solely in terms of the lattice operator, one needs to match eigenvectors between the standard continuum operator, and the continuum equivalent of the lattice operator. This is achieved by the operator $Z$ in the definition of the blocking. But there is no continuum equivalent to the zero mode doublers. We can neglect this, by using the infinite number of variables in the continuum to hide one additional lattice mode, and this is what I did in~\cite{Cundy:2009ab}, and it is good enough if we only wish to consider the construction of the overlap operator and lattice chiral symmetry. For this work, I shall continue to use this approach, despite problems which will arise in the discussion of $\mathcal{CP}$ symmetry. In a subsequent article, I will explore the possibility of introducing a second `doubler' field in the continuum (with a mass of the order of the cut-off), which will allow a more careful analysis of the difficulties that arise from the zero mode doublers.

These $\gamma_5$ operators do not commute (except at $\eta_1 = \eta_2$), but satisfy the following relations\footnote{See appendix \ref{app:glgr} for proofs of this and other results given in this discussion},
\begin{align}
\gamma_R^{(\eta_1)}\gamma_R^{(\eta_2)} =& \gamma_R^{2 \eta_1 - \eta_2} \gamma_R^{\eta_1}\nonumber\\
 \gamma_L^{(\eta_2)}\gamma_L^{(\eta_1)} =& \gamma_L^{\eta_1} \gamma_L^{2 \eta_1 - \eta_2}. 
\end{align}
%OK, need to introduce the regularisation of $(\gamma_5 - \sign(K))$.
Additional Ginsparg-Wilson equations can also be found by addition; for example $\gamma_R =n \sum_{i} c_i \gamma_R^{(\eta_i)}$, $\gamma_L =n \sum_{i} c_i \gamma_L^{(\eta_i)}$, for arbitrary $c_i$ and $\eta_i$, and the coefficient $n$, which will be a function of $D^{\dagger}D$, is chosen to ensure that $\gamma_R^{2} = \gamma_L^2 = 1$. However, I have not yet found a renormalisation group blocking which constructs these combined operators; and it is therefore unclear that the different Ginsparg-Wilson symmetries are connected to each other by the renormalisation group.

These $\gamma_5$-matrices satisfy the following properties (the proofs are either by inspection or can be found in appendix \ref{app:glgr}):
\begin{enumerate}
\item \textbf{Continuum limit:}
In the continuum limit ($a \rightarrow 0$), $\theta = \pi/2$, and $\gamma_L^{(\eta)} = \gamma_R^{(\eta)} = \gamma_5$.
\item \textbf{Hermiticity:} $\gamma_R^{(\eta)} = (\gamma_R^{(\eta)})^{\dagger}$. $\gamma_L^{(\eta)} = (\gamma_L^{(\eta)})^{\dagger}$.\footnote{More precisely, this depends on how the zero mode doublers are regulated in the matrix sign function of $\gamma_5(D^{\dagger} - D)$.}
\item \textbf{Unitarity:}  $(\gamma_R^{(\eta)})^2 = (\gamma_L^{(\eta)})^2 = 1$.
\item \textbf{Ginsparg-Wilson chiral symmetry}: $\gamma_L D + D \gamma_R = 0$.
The chiral transformations associated with this Ginsparg Wilson equation are
\begin{align}
\psi \rightarrow& e^{i\epsilon \gamma_R}\psi; & \overline{\psi} \rightarrow \overline{\psi}  e^{i\epsilon \gamma_L},\label{eq:chisym}
\end{align}
and for an infinitesimal transformation the change in the fermionic action is
\begin{gather}
\Delta S = i \epsilon \overline{\psi}(\gamma_L D + D \gamma_R) \psi = 0.
\end{gather}
The topological charge can be defined by considering the Jacobian of the chiral transformation
\begin{gather}
Q = -\frac{1}{2}\Tr(\gamma_R + \gamma_L)
\end{gather}
\item $\mathcal{CP}$\textbf{-symmetry} (see appendices \ref{app:CP} and \ref{app:eigenvalues}):
\begin{align}
\mathcal{CP:} \gamma_R^{(\eta)} =& -W(\gamma_L^{(-\eta)})^TW^{-1} \nonumber\\
\mathcal{CP:} \gamma_L^{(\eta)} =& -W(\gamma_R^{(-\eta)})^TW^{-1} \label{eq:glgrcp}
\end{align}
\end{enumerate}
The behaviour of the zero modes $\phi_0$ and their partners $\phi_2$ under  $\mathcal{CP}$ is discussed in appendix \ref{app:eigenvalues}, where it is shown that $\mathcal{CP}$ preserves the eigenvectors, $\mathcal{CP}:\ket{\phi_0} = \ket{\phi_0^{CP}} =  W^{-1}\ket{\phi_0}$, but the chirality of the eigenvectors is switched, $\bra{\phi_0}\gamma_5\ket{\phi_0} = - \bra{\phi_0^{CP}}\gamma_5\ket{\phi_0^{CP}}$. This means that when $\mathcal{CP}$ is applied to the chiral Lagrangian, $\frac{1}{4}\overline{\psi}(1-\gamma_L^{(\eta)}) D (1+\gamma_R^{(\eta)})\psi$ the contribution of the zero mode doublers remain the same: if they contributed to the original Lagrangian, the zero mode doublers of the $\mathcal{CP}$ transformed Dirac operator will contribute to the transformed Lagrangian. The only other effect of $\mathcal{CP}$ symmetry on Dirac operators and blocking operators  is to switch $\eta \rightarrow - \eta$.

\section{Lattice $\mathcal{CP}$ and the Ginsparg Wilson relation}\label{sec:2a}
As stated in~\cite{Cundy:2009ab}, overlap lattice QCD can be derived directly from the renormalisation group. This presents a natural mechanism for deriving a lattice $\mathcal{CP}$ symmetry. The generating function is
\begin{align}
Z[U,J_0,\overline{J}_0, g_0] =  \int d\psi_0 d\overline{\psi}_0 d \psi_1^{(\eta)} d \overline{\psi}_1^{(\eta)} &e^{-\overline{\psi}_0 D_0 \psi_0 - \frac{1}{4g_0^2}F_{\mu\nu}^2 - \overline{J}_0 \psi_0 - \overline{\psi}_0 J}\nonumber\\ & e^{-(\overline{\psi}_1^{(\eta)} - \overline{\psi}_0(\overline{B}^{(\eta)})^{-1})\alpha({\psi}_1^{(\eta)} - ({B}^{(\eta)})^{-1}\psi_0)}. 
\end{align}
Given that the original action is invariant, applying $\mathcal{CP}$ gives,
\begin{align}
Z[U,J_0,\overline{J}_0, g_0] = & \int d\psi_0 d\overline{\psi}_0 d \psi_1^{(\eta)} d \overline{\psi}_1^{(\eta)} e^{-\overline{\psi}_0 D_0 \psi_0 - \frac{1}{4g_0^2}F_{\mu\nu}^2 - \overline{J}_0 \psi_0 - \overline{\psi}_0 J} \nonumber\\
&e^{-({(\overline{\psi}^{(\eta)}_1)^{CP}} - {\overline{\psi}_0^{CP}}((\overline{B}^{(\eta)})^{-1})^{CP})\alpha^{CP}({({\psi}_1^{(\eta)})^{CP}} - ((B^{(\eta)})^{-1})^{CP}{\psi_0^{CP}})}. 
\end{align}
Where, as shown in equation (\ref{eq:CPfromB}) of appendix \ref{app:eigenvalues}, $\mathcal{CP}:B^{(\eta)} = (B^{(\eta)})^{CP} =  (\overline{B})^{(-\eta)}$.

The generating function is invariant under $\mathcal{CP}$ if $\alpha^{CP} = \alpha$, and
\begin{align}
\ket{(\psi_1^{(\eta)})^{CP}} =& -W^{-1}(\bra{\overline{\psi}_0}(\overline{B}^{(-\eta)})^{-1})^{T} = -W(\bra{\overline{\psi}_1^{(\eta)}}\overline{B}^{(\eta)} (\overline{B}^{(-\eta)})^{-1})^T \nonumber\\
\bra{(\overline{\psi}_1^{(\eta)})^{CP}} =& (({B}^{(-\eta)})^{-1}\ket{{\psi_0}})^{T}W^{-1} = ((B^{(-\eta)})^{-1} B^{(\eta)} \ket{\psi_1^{(\eta)}})^TW^{-1} .\label{eq:19}
\end{align}
I define
\begin{align}
(B^{(-\eta)})^{-1}B^{(\eta)} =& \hat{\gamma}_R^{(-1-\eta)}\gamma_5\nonumber\\
\overline{B}^{(\eta)} (\overline{B}^{(-\eta)})^{-1} = &\gamma_5\hat{\gamma}_R^{(-\eta-1)}.\label{eq:dontknowwhattocallthis}
\end{align}
For the non-zero eigenvectors of $D$, equation (\ref{eq:dontknowwhattocallthis}) follows from the eigenvalue decomposition of $B$ and $\overline{B}$ used to derive $\hat{\gamma}_L$ and $\hat{\gamma}_R$. For example, if we match the eigenvalues so that $\lambda = \lambda_0$,
\begin{align}
\overline{B}^{(\eta)} =& \left(\frac{\lambda_2}{\lambda_0}\right)^{(1-\eta)/2}\left(\begin{array}{l l}
\cos((1-\eta)(\theta - \pi/2)/2)&\sin((1-\eta)(\theta - \pi/2)/2)\\
-\sin((1-\eta)(\theta - \pi/2)/2)&\cos((1-\eta)(\theta - \pi/2)/2)\end{array}\right) \nonumber\\
\overline{B}^{(\eta)} (\overline{B}^{(-\eta)})^{-1} =& \left(\begin{array}{l l}
\cos((\eta)(\theta - \pi/2)/2)&\sin((\eta)(\theta - \pi/2)/2)\\
-\sin((\eta)(\theta - \pi/2)/2)&\cos((\eta)(\theta - \pi/2)/2)\end{array}\right)
\nonumber\\
=&\gamma_5 \overline{B}^{(-1-\eta)}\gamma_5(\overline{B}^{(-1-\eta)})^{-1}  = \gamma_5\hat{\gamma}_R^{(-\eta-1)}.\label{eq:BB}
\end{align}
Equation (\ref{eq:dontknowwhattocallthis}) can also be extended to the zero modes and zero mode doublers; although except when $\eta$ is an even integer it will not be possible to give a local closed form that describes these operators both for the zero mode doublers and the non-zero modes. Combining equations (\ref{eq:19}) and (\ref{eq:dontknowwhattocallthis}) gives
\begin{align}
\hat{\gamma}_R^{(-\eta - 1)} =& \left[\gamma_5 \cos(\eta(\theta - \pi/2)) - \sign(\gamma_5(D-D^{\dagger}))\sin(\eta(\theta - \pi/2))\right](1-\ket{\phi_2}\bra{\phi_2}) + \gamma_5 \ket{\phi_2}\bra{\phi_2}\nonumber\\
\mathcal{CP}: \overline{\psi}^{(\eta)} = & \left(\hat{\gamma}_R^{(-\eta - 1)} \gamma_5\psi^{(\eta)}\right)^T
\nonumber\\
\mathcal{CP}: {\psi}^{(\eta)} = &- \left(\overline{\psi}^{(\eta)}  \gamma_5 \hat{\gamma}_R^{(-\eta - 1)}\right)^T
\end{align}
It is not necessary that these transformations are local since the fermion fields themselves need not be local. After all, even in the continuum the parity operation is not local. It is, however, necessary that lattice $\mathcal{CP}$ smoothly reduces to continuum $\mathcal{CP}$ as the limit of zero lattice spacing is taken.
Using equations (\ref{eq:resultneededinCP}) and (\ref{eq:secondresultneededinCP}), and noting that the $\mathcal{CP}$ transformation of the zero mode doublers has the opposite chirality, it can be shown that both the standard and chiral gauge fermionic Lagrangians are invariant under lattice $\mathcal{CP}$:
\begin{align}
\mathcal{CP}:&\overline{\psi} D \psi = \left(\overline{\psi} \gamma_5 \hat{\gamma}_R^{(-\eta - 1)} D\hat{\gamma}_R^{(-\eta - 1)} \gamma_5 \psi\right)^T \nonumber\\
&= \overline{\psi} D \psi\nonumber\\
 \mathcal{CP}:&\frac{1}{2}\overline{\psi} D(1+\gamma_R^{(\eta)}) \psi \nonumber\\
&=\frac{1}{2} \left(\overline{\psi} \gamma_5 \hat{\gamma}_R^{(-\eta - 1)} (1-\gamma_L^{(-\eta)})D\hat{\gamma}_R^{(-\eta - 1)} \gamma_5 \psi\right)^T \nonumber\\
&= \overline{\psi}  D(1+\gamma_5 \hat{\gamma}_R^{(-\eta - 1)}\gamma_R^{(-\eta)}\hat{\gamma}_R^{(-\eta - 1)} \gamma_5) \psi  \nonumber\\
&= \frac{1}{2}\overline{\psi} D(1+\gamma_R^{(\eta)})\psi\label{eq:psiCP}
\end{align}
This formulation generalises the approach suggested in~\cite{Igarashi:2009kj}.

This means that any of these primary formulations of a lattice chiral gauge theory (as opposed to the secondary formulations derived from combinations of $\gamma_{L,R}^{(\eta)}$ derived by adding two or more solutions together) are $\mathcal{CP}$ invariant if the correct lattice $\mathcal{CP}$ symmetry is used.

This approach can be easily extended to the additional Ginsparg-Wilson operators based on the overlap operator described in~\cite{Cundy:2008cs} which are related to the overlap operator by additional Ginsparg-Wilson transformations~\cite{Cundy:2008cn}.
\section{The effects of $\mathcal{CP}$ violation}\label{sec:2b}
The effects of $\mathcal{CP}$ violation in a chiral gauge theory were discussed in~\cite{Fujikawa:2002vj}. They discovered three effects: (1) an overall constant phase in the generating functional; (2) an overall constant coefficient in the fermion generating functional; and (3) a shift in the quark propagator in external lines and when connected to Yukawa vertices. I shall discuss this last point, which could, in principle, lead to non-localities when coupled to a Higgs field, in section \ref{sec:8b}. In this section, I discuss the measure of the chiral gauge theory and the fermion propagator.

\subsection{Construction of Weyl fermion fields}\label{sec:4.0}
Unlike the standard action, the functional form of the lattice Weyl action will depend on the blocking used to convert from the continuum to the lattice. With this caveat, a lattice Weyl action can be easily constructed from a continuum Weyl action:
\begin{align}
\frac{1}{2}\overline{\psi}_0 D_0 (1+\gamma_5)\psi_0 =& \frac{1}{2}\overline{\psi}_1 \overline{B}^{(\eta)} D_0 B^{(\eta)}(1+ (B^{(\eta)})^{-1}\gamma_5B^{(\eta)})\psi_1\nonumber\\
=&  \frac{1}{2}\overline{\psi}_1^{(\eta)} D (1+\gamma_R^{(\eta)}) \psi_1^{(\eta)}.
\end{align}
This is therefore the particular lattice Weyl action associated each blocking.

The prescription to consider changes in the measure caused by infinitesimal changes in the gauge field was developed by Martin Luscher in \cite{Luscher:1998du,Luscher:1999un}. We can select a set of basis vectors $v_0$ and $\overline{v}_0$ which satisfy $\hat{\gamma}_R^{\eta} v_0 = v_0$ and $\overline{v}_0 \hat{\gamma}_L^{\eta}  = \overline{v}_0$. Under an infinitesimal change in the gauge field,
\begin{align}
\delta_{\xi}U_{\mu}(x) =& \xi(x)U_{\mu}(x);&\xi(x) =& \xi^a(x) T^a,
\end{align}
the Jacobian for the change in the measure for the physical fields is
\begin{gather}
e^{-i\mathfrak{L}_{\xi}^{(\eta)}},\nonumber
\end{gather}
where
\begin{gather}
\mathfrak{L}_{\xi}^{(\eta)} = i\sum_i [(v_0,\delta_{\xi} v_0) + (\delta_{\xi} \overline{v}_0,\overline{v}_0)],\label{eq:Ljac}
\end{gather}
and $\delta_{\xi} v_i$ is the infinitesimal change in the basis vector associated with the change in the gauge field. Then, to consider the change in the measure for a change in the gauge field and the basis one computes the Wilson line across a trajectory in the space of gauge fields
\begin{align}
W=&e^{i\int dt \mathfrak{L}^{(\eta)}_{\xi}}& \xi(x) =& \partial_t (U^t_{\mu}(x)) (U^t_{\mu}(x))^{-1}.\label{eq:wilsonline}
\end{align}

\subsection{The fermionic measure}\label{sec:4.1}

The eigenvalues of the overlap operator come in pairs, either non-zero pairs $\phi_+^{(\eta)}$ and $\phi_-^{(\eta)}$, where $H^2 \phi_+^{(\eta)} = \lambda^2\phi_+^{(\eta)}; H^2 \phi_-^{(\eta)} = \lambda^2\phi_-^{(\eta)}$, or the zero modes and their partners. Here $\phi_+^{(\eta)}$ and $\phi_-^{(\eta)}$ are chosen to be eigenvectors of $\gamma_R^{(\eta)}$, while we can equally construct two vectors  $\overline{\phi}_+^{(\eta)}$ and $\overline{\phi}_-^{(\eta)}$ which are eigenvectors of $\gamma_L^{(\eta)}$. The fermion fields can be written in terms of these vectors,
\begin{align}
\ket{\psi} =& \sum c_{-i}^{(\eta)} \ket{\phi_{-i}^{(\eta)}} +\sum c_{+i}^{(\eta)} \ket{\phi_{+i}^{(\eta)}} + \sum c_{0i}\ket{\phi_{0i}} + \sum c_{2i}\ket{\phi_{2i}}  \nonumber\\
\bra{\overline{\psi}} =&\sum \overline{c}^{(\eta)}_{-i} \bra{\overline{\phi}_{-i}^{(\eta)}} +\sum \overline{c}_{+i}^{(\eta)} \bra{\overline{\phi}_{+i}^{(\eta)}} + \sum \overline{c}_{0i}\bra{\overline{\phi}_{0i}} + \sum \overline{c}_{2i}\bra{\overline{\phi}_{2i}}\label{eq:21} 
\end{align}
which allows the measure to be defined in terms of $c^{(\eta)}$ and $\overline{c}^{(\eta)}$. Each basis, which I shall label by $\omega$, is only defined up to a phase $e^{i\theta[\omega,U,\eta]}$. An infinitesimal change in the projection operators, whether from the gauge field or a change in $\eta$, will induce an infinitesimal change in the eigenvectors $\delta \phi^{(\eta)}$ unless there is a change in the topological index, and will change the fermion variables from $c$ to $c'$, where, for example,
\begin{align}
\psi =& \sum c_i \phi_i = \sum c'_i (\phi^{(\eta)}_i + \delta \phi^{(\eta)}_i),\nonumber\\
 c_i  =& \sum c'_j (\delta_{ij} + (\phi^{(\eta)}_j,\delta \phi^{(\eta)}_i).
\end{align}
Using the result $(\phi_i,\delta\phi_i)=0$, the Jacobian for this transformation is $\det(\delta_{ij} + (\phi_j,\delta\phi_i)) = \exp(\Tr\log(\delta_{ij} +(\phi_j,\delta\phi_i))) = \exp(\Tr (\phi_j,\delta\phi_i)) = 1$ (see section \ref{sec:7}), as long as the change in the gauge field induces only an infinitesimal change in the fermion field, which will be the case as long as the topological index remains constant. Thus the measure remains invariant for all continuous changes of a gauge field within a topological sector, including gauge transformations.  This result will also hold for infinitesimal changes in $\eta$. There is no discontinuity at $\eta=0$ because the contribution of the zero modes and their partners to the fermionic measure is independent of $\eta$. 

%Thus there is no change in the Jacobian for the fermionic measure when the sign of $\eta$ is flipped. This means that under $\mathcal{CP}$ the fermionic measure remains constant. 

I need to show that the measure is invariant under both gauge transformations and $\mathcal{CP}$. In~\cite{Cundy:2009ab}, I discussed the question of gauge invariance of the measure for $\eta = 0$ in the absence of zero modes and their partners, and I shall generalise that argument in section \ref{sec:7}. Here, I shall concentrate on the $\mathcal{CP}$ invariance of the measure. In~\cite{Fujikawa:2002vj}, an analysis based on L\"uscher's approach concluded that is was possible to construct a basis where the measure was $\mathcal{CP}$-invariant.\footnote{Some care needs to be taken with this approach because the construction of the Wilson line in equation (\ref{eq:wilsonline}) is only valid if the gauge field does not cross the topological index boundary, since the basis changes discontinuously when there is a change in the topological index. Except for trivial topology, $U$ and $U^{CP}$ will be in different topological sectors.} Here, I will take a different approach based on a specific construction of the basis. Here the basis vectors $v_R^{(\eta)}$ and $v_L^{(\eta)}$ are constructed as the eigenvectors of the projectors $P^{(\eta)}_R = (1+\hat{\gamma}_R^{(\eta)})/2$ and $P^{(\eta)}_L = (1-\hat{\gamma}^{(\eta)}_L)/2$ so that they satisfy $P^{(\eta)}_R v_R^{(\eta)} = v_R^{(\eta)}$, $P^{(\eta)}_L v_L^{(\eta)} = v_L^{(\eta)}$. 

Using ${\gamma}_L^{(\eta)} = \hat{\gamma}_R^{(-\eta - 1)}{\gamma}_R^{(\eta)}\hat{\gamma}_R^{(-\eta - 1)}$, and with the Weyl fermion fields, $\psi_+$ and $\psi_-$ restricted to the chiral sector containing zero modes\footnote{The fermion fields in the opposite chiral sector can be treated in the same way.}, the basis given in equation (\ref{eq:21}) can be defined as
\begin{align}
\ket{\psi_+} =& \sum_i c_i^{(\eta)} \ket{\phi^{(\eta)}_{+,i}} + c_0 \ket{\phi_0}\nonumber\\
\bra{\overline{\psi}_-} =& \sum_i \overline{c}_i^{(\eta)} \bra{\phi^{(\eta)}_{-,i}}\hat{\gamma}_R^{(-\eta - 1)}  + \overline{c}_0 \bra{\phi_0},\label{eq:basis7}
\end{align}
with a measure
\begin{gather}
\prod_i c_i \overline{c}_i c_0 \overline{c}_0. 
\end{gather} 
The basis used in the $\mathcal{CP}$ transformed action will be constructed from the eigenvectors of $\hat{\gamma}_L^{(-\eta)} = \gamma_5 \hat{\gamma}_R^{(\eta)} \gamma_5$ and $\hat{\gamma}_R^{(-\eta)} = \gamma_5\hat{\gamma}_R^{(\eta - 1)} \hat{\gamma}_R^{(\eta)}\hat{\gamma}_R^{(\eta -1)} \gamma_5$, giving 
\begin{align}
-\ket{\psi^{CP}_+} = & \sum_i c_i^{CP} W^{-1}\left(\bra{\phi^{(\eta)}_{+,i}}\hat{\gamma}_R^{(\eta - 1)}\gamma_5\right)^T + c_0^{CP} \ket{\phi_0^{CP}}\nonumber\\
\bra{\overline{\psi}^{CP}_-} = & \sum_i \overline{c}^{CP}_i  \left(\gamma_5\ket{\phi^{(\eta)}_{-,i}}\right)^T W + \overline{c}_0^{CP} \bra{\phi_0^{CP}}\label{eq:basisCP7}
\end{align}

The basis is constructed in terms of the eigenvectors of the Dirac operator. In appendix \ref{app:eigenvalues}, I show that, for each non-zero eigenvalue pair, the eigenvectors of $\hat{\gamma}_R^{(\eta)}$ can expressed in terms of the non-zero eigenvalues of $H=\gamma_5 + \sign(K)$, 
\begin{align}
\ket{\phi_+^{(\eta)}} =&\cos \alpha^{(\eta)} \ket{H_+} + \sin \alpha^{(\eta)} \ket{H_-}\nonumber\\
\ket{\phi_-^{(\eta)}} =&\cos \alpha^{(\eta)} \ket{H_-} - \sin \alpha^{(\eta)} \ket{H_+}\nonumber\\
\alpha^{(\eta)} =&(\theta + (\eta + 1) (\pi/2 - \theta))/2\label{eq:phialpha}
\end{align} 
where
\begin{align}
(\gamma_5 + \sign(K))|H_{+i}\rangle =& \lambda_i |H_{+i}\rangle\nonumber\\
(\gamma_5 + \sign(K))|H_{-i}\rangle =& -\lambda_i |H_{-i}\rangle.
\end{align}
For simplicity, I will usually suppress the eigenvalue index.
As shown in appendix \ref{app:eigenvalues}, under $\mathcal{CP}$, the eigenvectors transform as
\begin{align}
\mathcal{CP}:\ket{H_+[U]}\rightarrow \ket{H_+^{CP}} =& -W^{-1}(\bra{H_-[U^{CP}]})^T\nonumber\\
\mathcal{CP}:\ket{H_-[U]}\rightarrow \ket{H_-^{CP}} =& W^{-1}(\bra{H_+[U^{CP}]})^T.\label{eq:CPeigenvectors}
\end{align}
The zero modes and their partners transform as
\begin{align}
\mathcal{CP}:\ket{\phi_2[U]} \rightarrow \ket{\phi_2^{CP}} = & W^{-1}(\bra{\phi_2[U^{CP}]})^T\nonumber\\
\mathcal{CP}:\ket{\phi_0[U]} \rightarrow \ket{\phi_0^{CP}} = & W^{-1}(\bra{\phi_0[U^{CP}]})^T,
\end{align}
where $|\phi_2\rangle[U^{CP}] = |\phi_0\rangle [U]$. Putting together equations (\ref{eq:phialpha}) and (\ref{eq:CPeigenvectors}) and that $\mathcal{CP}$ transforms $\mathcal{CP}:\eta \rightarrow -\eta$, we have
\begin{align}
\mathcal{CP}:\ket{\phi_+^{(\eta)}}[U] =&\ket{\phi_+^{(\eta)}[U]^{CP}}= - W^{-1}\bra{\phi_-^{(-\eta)}}^T [U^{CP}])\nonumber\\
\mathcal{CP}:\ket{\phi_-^{(\eta)}}[U] =&\ket{\phi_-^{(\eta)}[U]^{CP}}=  W^{-1}\bra{\phi_+^{(-\eta)}}^T[U^{CP}].
\end{align}
Inserting the $\mathcal{CP}$ transformations of $\psi$ and $\overline{\psi}$ given in equation (\ref{eq:psiCP}) into equation (\ref{eq:basisCP7}) gives
\begin{align}
\bra{\overline{\psi}_-} = & -\sum_i c_i^{CP} \bra{\phi^{(\eta)}_{-,i}} - c_0^{CP} \bra{\phi_0}\nonumber\\
\ket{\psi_+} = & \sum_i \overline{c}^{CP}_i  \hat{\gamma}_R^{(\eta - 1)}\ket{\phi^{(\eta)}_{+,i}} + \overline{c}_0^{CP} \ket{\phi_0}.
\end{align}
Comparing this with equation (\ref{eq:basis7}) gives
the measure for the transformation as $(\det(\gamma_R^{(\eta-1)}))^2 = 1$. 
%Can do this better by writing it as  \ket{\phi_+i}\bra{\phi_{+j}\gamma_R \ket{\phi_+i} and showing that \bra{\phi_{+j}\gamma_R \ket{\phi_+i}=\delta_{ij}

Therefore the measure of the chiral gauge field is invariant under $\mathcal{CP}$. 
\subsection{The generating function}

Observables are defined with respect to the generating functional of fermionic Green's functions. We can write
\begin{align}
&\langle \psi(x_1) \psi(x_2) \ldots\overline{\psi}(y_1)\overline{\psi}(y_2)\ldots \rangle_U = \lim_{\chi\rightarrow0}\lim_{\overline{\chi}\rightarrow0}\left(\frac{\partial}{\partial{\overline{\chi}(x_1)}}\frac{\partial}{\partial{\overline{\chi}(x_2)}} \frac{\partial}{\partial{\chi(y_1)}}\frac{\partial}{\partial{\chi(y_2)}}\ldots\right) \log Z_F\nonumber\\
&Z_F = \int d\psi d\overline{\psi} Z_F^{\omega}\nonumber\\
&Z_F[\omega,U,\eta,\chi,\overline{\chi},\psi,\overline{\psi}]= e^{i\theta[\omega,U,\eta]}
e^{-\frac{1}{2}\overline{\psi} D (1+\gamma_R^{(\eta)} )\psi - \frac{1}{2}\overline{\psi}(1-\gamma_L^{(\eta)}) \chi - \frac{1}{2}\overline{\chi}(1+\gamma_R^{(\eta)})\psi}.\nonumber\\
&=e^{i\theta[\omega,U,\eta]}e^{-\frac{1}{2} (\overline{\psi}+ \overline{\chi}D^{-1}(1-\phi_0\phi_0^{\dagger})) D(1+\gamma_R^{(\eta)} )(\psi + (1-\phi_0\phi_0^{\dagger})D^{-1}\chi) + \frac{1}{2}\overline{\chi}  (1-\phi_0\phi_0^{\dagger}) (1+\gamma_R^{(\eta)} )D^{-1}\chi}\nonumber\\
& \phantom{space}e^{\frac{1}{2}(\overline{\psi},\phi_0)(\phi_0,(1-\gamma_L^{(\eta)})\chi) + \frac{1}{2}(\overline{\chi},(1+\gamma_R^{(\eta)})\phi_0)(\phi_0, \psi)}\nonumber\\
&=e^{i\theta[\omega,U,\eta]}\det(D(1+\gamma_R^{(\eta)})(1-\phi_0\phi_0^{\dagger})) (\overline{\chi}\frac{1}{2}(1+\gamma_R^{(\eta)}),\phi_0)(\phi_0\frac{1}{2}(1-\gamma_L^{(\eta)}),\chi)e^{\overline{\chi} G^{(\eta)} \chi},\label{eq:gf}
\end{align}  
where $Z_F^{\omega}$ is the generating functional defined in a particular basis of the gauge fields, and the propagator $G^{(\eta)}$ satisfies
\begin{gather}
G^{(\eta)} D= \frac{1}{2}(1+\gamma_R^{(\eta)})(1 - \phi_0 \phi_0^{\dagger}).
\end{gather} 
Under $\mathcal{CP}$, $\eta \rightarrow -\eta$ and the zero mode $\phi_0\rightarrow \phi_0^{CP}$. The generating functional is given by
\begin{align}
{Z_F^{\omega}}^{CP}&[\omega,U^{CP},-\eta,-W\overline\gamma_R^{(-\eta-1)}\gamma_5{\chi}^T,{\chi}^T\gamma_5\gamma_R^{(-\eta - 1)} W,\overline{\psi},\psi] =\nonumber\\
&e^{i\theta[\omega,U^{CP},-\eta]}
e^{-\frac{1}{2}\overline{\psi}\gamma_5\gamma_R^{(-\eta-1)} D (1+\gamma_R^{(-\eta)} )\gamma_R^{(-\eta - 1)}\gamma_5\psi}\nonumber\\& \phantom{space}e^{- \frac{1}{2}\overline{\psi}\gamma_5\gamma_R^{(-\eta-1)}(1-\gamma_L^{(-\eta)}) \gamma_R^{(-\eta - 1)}\gamma_5\chi - \frac{1}{2}\overline{\chi}\gamma_5\gamma_R^{(-\eta-1)}(1+\gamma_R^{(-\eta)})\gamma_R^{(-\eta - 1)}\gamma_5\psi}\nonumber\\
=&e^{-\frac{1}{2} (\overline{\psi}\gamma_5\gamma_R^{(-\eta-1)}+ \overline{\chi}\gamma_5\gamma_R^{(-\eta-1)}D^{-1}(1-\phi_0^{CP}(\phi_0^{CP})^{\dagger})) D(1+\gamma_R^{(-\eta)} )(\gamma_R^{(-\eta - 1)}\gamma_5\psi + (1-\phi^{CP}_0(\phi^{CP}_0)^{\dagger})D^{-1}\gamma_R^{(-\eta - 1)}\gamma_5\chi)}\nonumber\\
&\phantom{space} e^{i\theta[\omega,U^{CP},-\eta]}e^{\frac{1}{2}\overline{\chi}\gamma_5\gamma_R^{(-\eta-1)}  (1-\phi_0\phi_0^{\dagger}) (1+\gamma_R^{(-\eta)} )D^{-1}\gamma_R^{(-\eta - 1)}\gamma_5\chi}\nonumber\\
& \phantom{space}e^{\frac{1}{2}(\overline{\psi},\gamma_5\gamma_R^{(-\eta-1)}\phi_0^{CP})(\phi_0^{CP},(1-\gamma_L^{(-\eta)})\gamma_R^{(-\eta - 1)}\gamma_5\chi) + \frac{1}{2}(\overline{\chi},\gamma_5\gamma_R^{(-\eta-1)}(1+\gamma_R^{(-\eta)})\phi_0^{CP})(\phi_0^{CP}, \gamma_R^{(-\eta - 1)}\gamma_5\psi)}\nonumber\\
=&e^{i\theta[\omega,U,\eta]}\det(D(1+\gamma_R^{(\eta)})(1-\phi_0(\phi_0)^{\dagger})) (\overline{\chi}\frac{1}{2}(1+\gamma_R^{(\eta)}),\phi_0^{CP})\nonumber\\
&\phantom{some new space and really new and big space}(\phi_0\frac{1}{2}(1-\gamma_L^{(\eta)}),\chi)e^{\overline{\chi} (G^{(\eta)})^{CP} \chi},\label{eq:gfcp}
\end{align}
with the propagator
\begin{align}
D (G^{(\eta)})^{CP} =& \frac{1}{2}\gamma_5\gamma_R^{(-1-\eta)}(1-\ket{\phi_0^{CP}}\bra{\phi_0^{CP}})(1+(\gamma_R^{(-\eta)}))\gamma_R^{(-1-\eta)}\gamma_5\nonumber\\
=&\frac{1}{2}(1-\ket{\phi_0}\bra{\phi_0})(1+(\gamma_R^{(\eta)})).\nonumber\\
=&D (G^{(\eta)})
\end{align}
Thus, comparing equations (\ref{eq:gfcp}) and (\ref{eq:gf}) with the correct lattice formulation of $\mathcal{CP}$ symmetry, the chiral gauge propagator and generating functional are invariant under $\mathcal{CP}$.
    
\section{Locality of the $\gamma_5$ operators}\label{sec:3}
Clearly, for this work to be valid, $\gamma_L^{(\eta)}$ and $\gamma_R^{(\eta)}$ must be local. I outlined an argument why $\tilde{\gamma}_{L,R}$ should be local at $\eta=\pm 1$ in~\cite{Cundy:2009ab}, although I did not prove it for the general case. Here I argue that these operators are local for odd integer $\eta$.

The potential problem in ${\gamma}_R$ is the presence of eigenvalues of $D$ at exactly 2, where both the denominator and numerator in the second term of equation (\ref{eq:glgr}) are zero, or (equivalently) the kernel of the sign function is zero. At odd integer $\eta$, this term does not contribute so there is not this difficulty. It should first be observed that $\gamma_L$ and $\gamma_R$ themselves are finite: this is obvious from $\gamma_L^2 = 1$ and $\gamma_R^2 = 1$.  However, the presence of the square roots leads to concerns that there may be branch cuts in the Fourier transform of the Dirac operators. 

I shall also only consider ${\gamma}_R$ in the topological trivial sector (in the presence of zero modes, the argument following is invalid since I am neglecting the outer sign function of $\gamma_R^{(\eta)}$).
I note that, using equation (\ref{eq:notabene}),
\begin{align}
{\gamma}_R^{(\eta -2)} = \gamma_5{\gamma}_R^{(1)}{\gamma}_R^{(\eta)} = (1-D)\gamma_R^{(\eta)}.
\end{align}
If we assume, following already well-established results, that $D$ and $D^{\dagger}$ are local on a smooth enough gauge field that the overlap operator is local, it is sufficient to consider the cases when $\eta = 1$ and $\eta = 0$ and the locality (or otherwise) of $\gamma_R^{(\eta)}$ for all integer $\eta$ follows by induction. The locality (or non-locality) of $\gamma_L$ can then be demonstrated using $\gamma_L^{(\eta)} = \gamma_5 \gamma_R^{(-\eta)} \gamma_5$.

\subsection{The Paley-Wiener theorem}

If we consider a function $F(x)$ in the continuum, then we can construct the Fourier transform $\tilde{F}(p)$ according to
\begin{align}
\tilde{F}(p) =& \int_{-\infty}^{\infty} d^4x F(x) e^{i(p,x)}\nonumber\\
F(x) = & \frac{1}{(2\pi)^4} \int_{-\infty}^{\infty} d^4p \tilde{F}(p) e^{-i(p,x)}.
\end{align}
If $F$ possesses an O(4) rotational symmetry, then it suffices to only consider $x$ in one particular direction, for example along $p_0$
\begin{align}
F(x) = & \frac{1}{(2\pi)^4} \int_{-\infty}^{\infty} d^4p \tilde{F}(p_0,p_1,p_2,p_3) e^{-ip_0x}.
\end{align}
For $x>0$, the integral can be completed in the complex plane around the lower half circle as long as $\tilde{F}(p)$ is finite on this circle and zero at $p = \pm \infty$. If $\tilde{F}(p)$ is analytic along this contour, which, most importantly, includes being analytic along the real axis, then the integral over $p$ is then given by the sum over the residues of $\tilde{F}(p)$ and the integral around the branch cuts in the lower half complex plane. If there are no branch cuts,  and if the poles are at $\pi_i = \pi_i^{re} \pm i \pi_i^{im}$ ($\pi_i^{re},\pi_i^{im}$ both real and $\pi_i^{im} > 0$; $\pi_i$ will in general be a function of the other components of the momenta), then $F(x)$ will have the form,
\begin{gather}
F(x) = \sum_i \int d^3 p \alpha_i(\pi_i,p_1,p_2,p_3) e^{-\pi_i^{im}|x|}.\label{eq:pw1}
\end{gather}     
Assuming that when the remaining integral is calculated or estimated (for example, by using the method of steepest descent) it retains the exponential form (which will certainly occur if $\pi_i^{im}$ is constrained to be greater than some positive number $\beta$), then
\begin{gather}
|F(x)| < \alpha e^{-\beta |x|}\label{eq:pw2}
\end{gather}
 for some positive $\alpha$ and $\beta$ and $F$ is local. If there are additionally branch cuts as well as poles, from momentum $\pi_{j\;\text{start}}$ (with imaginary component closest to zero) to $\pi_{j\;\text{end}}$ then the functional form of equation (\ref{eq:pw1}) will no longer be valid, and instead we must use
\begin{align}
F(x) = \int d\theta_1 d\theta_2 d\theta_3 &\sin^2{\theta_1}\sin{\theta_2} \bigg[\sum_i \alpha_i(\pi_i,p_1,p_2,p_3) e^{-\pi_i^{im}|x|} +\nonumber\\ 
& \sum_j \int_{\pi_{j\;\text{end}}}^{\pi_{j\;\text{start}}} d\pi \alpha^j(\pi,p_1,p_2,p_3) e^{-i\pi x}\bigg].
\end{align}
Since, at each point along the integral, $ |\alpha^j(\pi,p_1,p_2,p_3) e^{-i\pi x}| < C' e^{-p_{j\;\text{start}} |x|}$ for some positive $C'$, the whole integral is smaller than $C e^{-\pi_{j\;\text{start}}|x|}$ and the equation (\ref{eq:pw2}) still holds as long as the branch cut does not cross the real axis. Thus for any function in the continuum where the Fourier transform is analytic along the real axis, the function itself is at least exponentially local. This is, of course, the Paley-Wiener theorem~\cite{Paley-Wiener}. 

%We are talking Fourier Series, not Fourier Transform . it's a lattice.
On the lattice, the momentum is bounded, $|p| < \pi/a$, so before the Paley-Weiner theorem can be applied it is necessary to transform  to a new momentum variable, $\hat{p}$ which is not bounded, for example using
\begin{gather}
 \hat{p}_{\mu} \frac{a}{2} = \tan\left(\frac{a}{2}{p}_{\mu}\right). 
\end{gather}
This gives
\begin{gather}
d p_{\mu} = d \hat{p}_{\mu} \frac{1}{1 + \left(\frac{a \hat{p}_{\mu}}{2}\right)^2}.\label{eq:latticepole}
\end{gather}
After this transformation of variables, the Paley-Wiener theorem can be derived in the same way; once again the result is that if $\tilde{F}(p)$ as analytic along the real axis for $-\pi/a<\hat{p}<\pi/a$, then the resulting operator $F(x)$ will be exponentially local or better. Equation (\ref{eq:latticepole}) indicates that any lattice operator for which this construction is valid will only be exponentially local, with a rate of decay inversely proportional to the lattice spacing, and cannot be ultra-local.

\subsection{Application to $\gamma_R^{(\eta)}$}
It is useful to consider two cases separately, when $\eta$ is an odd integer, and $\eta$ is an even integer. 

The definition of $\gamma_R^{(\eta)}$ is
\begin{align}
\gamma_R^{(\eta)} =& \gamma_5 \cos((1+\eta)\pi/2)\left[\cos((1+\eta)\theta) - \frac{D^{\dagger}-D}{2\sin2\theta} \sin((1+\eta)\theta)\right]+\nonumber\\
&\gamma_5 \sin((1+\eta)\pi/2)\left[\sin((1+\eta)\theta) + \frac{D^{\dagger}-D}{2\sin2\theta} \cos((1+\eta)\theta)\right],
\end{align}
where
\begin{align}
\cos(\theta) =& \sqrt{1-D^{\dagger}D/4}\nonumber\\
\sin(\theta) = & \sqrt{D^{\dagger}D/4}.
\end{align}

\subsubsection{$\eta$ odd integer}
If $\eta$ is an odd integer, then $\sin((1+\eta)\pi/2) = 0$. The locality of $\gamma_R^{(\eta)}$ can be proved using the well-known result given in lemma \ref{lemma:5.1}
\begin{lemma}\label{lemma:5.1}
For integer $n>0$, $\cos(2n\theta) = P_C^n(\cos^2 \theta)$ and $\sin(2n\theta) = \sin (2 \theta) P_S^{n-1}(\cos^2 \theta)$ where $P_C^n$ and $P_S^n$ are polynomials of order $n$.
\end{lemma}
\begin{proof}
From de Moivre's theorem,
\begin{align}
\cos (2n\theta) =& \Re\left[(\cos \theta + i \sin \theta)^{2n}\right]\nonumber\\
=& \Re\left[\sum_{m = 0}^{2n} \frac{(2n)!}{m! (2n-m)!} \cos^{m}\theta (i\sin\theta)^{2n-m} \right]\nonumber\\
=& \sum_{m = 0}^n \frac{(2n)!}{(2m)! (2n-2m)!} (-1)^m \cos^{2m}\theta (1-\cos^{2}\theta)^{n-m}\label{eq:cosntheta}
\end{align}
and the result immediately follows. Similarly, 
\begin{gather}
\sin(2n \theta) = -\sum_{m=0}^{n-1} (-1)^m\frac{(2n)!}{(2m+1)! (2n-2m -1)!} \cos\theta \sin\theta \cos^{2m}\theta (1-\cos^2\theta)^{n-m}.\label{eq:sinntheta}
\end{gather}
\qed
\end{proof}
Therefore,
\begin{align}
\gamma_R^{(\eta)} =& \gamma_5 \cos((1+\eta)\pi/2)\left[\cos((1+\eta)\theta) - \frac{D^{\dagger}-D}{2\sin2\theta} \sin((1+\eta)\theta)\right] \nonumber\\
=&\gamma_5 \cos((1+\eta)\pi/2) \left[P_C^{((\eta+1)/2)}(\cos^2\theta) - (D^{\dagger} - D) P_S^{((\eta-1)/2)}(\cos^2\theta)\right]\nonumber\\
=&\gamma_5 \cos((1+\eta)\pi/2) \left[P_C^{((\eta+1)/2)}(D^{\dagger}D/4) - (D^{\dagger} - D) P_S^{((\eta-1)/2)}(D^{\dagger}D/4)\right].
\end{align}
Given that a converging polynomial of a local function is itself local, it is clear that, for odd integer $\eta$, $\gamma_R^{(\eta)}$ is local on all gauge field configurations where the overlap operator is local for $\eta$ an odd integer. This includes the standard case where $\eta =  1$ or $\eta = -1$.
\subsubsection{$\eta$ an even integer}
It is enough to consider the case where $\eta = 0$; the locality or non-locality of other even $\eta$ will follow as outlined above. After Fourier transforming twice, 
\begin{align}
\gamma_5\gamma_R^{(0)}(x) = &\frac{1}{(2\pi^2)^4} \int_{-\infty}^{\infty} dt \int_{-\infty}^{\infty} d^3\hat{p}d \hat{p}_0  (1-D(\hat{p})/2) e^{i 2 x\arctan(a\hat{p}_0/2)/a }\nonumber\\
&\phantom{somespace} \frac{1}{1+ (a \hat{p}_0/2)^2} \prod_i\frac{1}{1+ (a \hat{p}_i/2)^2}\frac{1}{t^2 + 1-D^{\dagger} D/4}.
\end{align}
To solve this integral is rather challenging, but an exact solution is not necessary to draw out general features, and a qualitative argument is enough to establish whether or not it is local. I proceed by using a contour integration over $\hat{p}_0$, then examining the integral over $t$ using either another contour integration or the method of steepest descent to establish the general shape of the function.

The integral over $\hat{p}_0$ has simple poles at $\hat{p}_0 = \pm i2/a$, which will not affect this argument as the poles are imaginary, and at the solutions to $t^2 + 1-D^{\dagger}(\hat{p}_0) D(\hat{p}_0)/4 = 0$, which will in general be at  $\hat{p}_0 = {\pi}^{(0)}(t,\hat{p}_1,\hat{p}_2,\hat{p}_3)$ for complex $\pi^{(0)}$. Here I will assume that there is just one such pole, but the argument can be easily extended if there are several solutions. Focusing on the pole with positive imaginary part, we can write
\begin{align}
\frac{1}{t^2 + 1-D^{\dagger} D/4} =& \frac{1}{t} \left(\frac{1}{t + i\sqrt{1-D^{\dagger} D/4}} + \frac{1}{t - i\sqrt{1-D^{\dagger} D/4}}\right)\nonumber\\
=& -\frac{1}{it} \frac{1}{(\hat{p}_0 - \pi^{(0)})\frac{\partial \sqrt{1-D^{\dagger} D/4}}{\partial \hat{p}_0} } + \ldots.  
\end{align}
Closing the contour around the appropriate semi-circle and performing the contour integration over $\hat{p}_0$ gives
\begin{align}
\gamma_5\gamma_R^{(\eta)}(x) = \ldots - & \frac{1}{\pi(2\pi)^4} \int_{-\infty}^{\infty} dt   (1-D(\hat{p}^{(0)}(t,\hat{p}_1,\hat{p}_2,\hat{p}_3))/2) \nonumber\\
&e^{-R((t,\hat{p}_1,\hat{p}_2,\hat{p}_3)) x/a + i T(t,\hat{p}_1,\hat{p}_2,\hat{p}_3) x/a} \frac{1}{1+ (a \hat{\pi}^{0}(t,\hat{p}_1,\hat{p}_2,\hat{p}_3)/2)^2} \frac{1}{t\frac{\partial \sqrt{1-D^{\dagger} D/4}}{\partial \hat{p}}},\label{eq:tt}
\end{align}
where $R$ is a positive function and $T$ another function. The $\ldots$ contains the contributions from the simple poles at $\hat{p}_i = \pm 2i/a$ and any branch cuts caused by the square roots in the definitions of $\cos\theta$ and $\sin\theta$ and the branch cut from the sign function in the overlap operator. Note that these functions are independent of the lattice spacing, since they are constructed by substituting the poles of $4t^2 + 4- D^{\dagger}(a\hat{p}_0)D(a\hat{p}_0) = 0$ into $\arctan (a \hat{p}_0/2)$. If $\pi^{(0)}$ has a large imaginary part, then $e^{i2x \arctan(\pi^{(0)}/2)}$ gives an exponential decay in $x$ (it cannot be an exponential growth because the contour for the integration must be closed in the semi-circle in the complex plane where $\Re (i2x \arctan(p_0 \cos\theta_1/2)) < 0$), and the contribution to the integrand from this particular $\pi^{(0)}$ is local. A real $\pi^{(0)}$ can only occur at $t = 0$ since $D^{\dagger}(\pi^{(0)}D^{(0)}$ is real, positive and less than or equal to 4 for real momenta. Therefore $\gamma_5\gamma_R^{(\eta)}$ may be non-local only if the integral is dominated by the section where $t$ is small. We can freely expand around $t=0$, knowing that where this expansion breaks down the effects will not affect a discussion of the locality.

 We can write $D(\pi^{(0)}) = c^a(t)T^a + i \gamma_{\mu} d^a_{\mu}(t)T^a$, where $T^a$ contains the colour structure. First I diagonalise the Dirac spinor components of this matrix. The eigenvalues of this operator are the eigenvalues of $\pm \sqrt{c^2 + d^2 + f_{abc}i\gamma_{\mu} d_{\mu}^b c^c T^a}$, which, at $D^{\dagger} D = 4 + 4t^2$ and small $t$ means that $1-D(\hat{p}^0)/2$ has eigenvalues of order $2 - (O(t^2))$ or  $t^2/2 + O(t^4)$. These two cases can be considered separately:
\begin{itemize}
\item Eigenvalues of O($t^2$): For large enough $|x|$, the integrals over $\hat{p}_1$, $\hat{p}_2$, $\hat{p}_3$ and $t$ can be performed using the method of steepest descent, expanding around the minima of $R(\pi^{(0)})$. The integral over $t$ is performed first. Here we need to find the minima of 
\begin{gather}
2 x/a \Im(\arctan (a \pi^{(0)}/2) - \log t + \Re(\log(1+(a\pi^{(0)}/2)^2), 
\nonumber
\end{gather}
which are located at
\begin{gather}
t^{-1} = \left[ x \Im \left[\frac{1}{ 1 + (a \pi^{(0)}/2)^2}\right] + \Re \frac{a \pi^{(0)}}{1 + (a \pi^{(0)} /2)^2}\right]\frac{\partial \pi^{(0)}}{\partial{t}}.\label{eq:tsol}
\end{gather}
While it is difficult to solve this exactly, nonetheless it is possible to describe the solution qualitatively.
If there is a real solution to $1-D^{\dagger} D /4 = 0$, the troublesome point, where $\pi^{(0)}(t,p_1,p_2,p_3)$ is close to the real axis, is around $t = 0$. We can therefore expand $\pi^{(0)}(t,p_1,p_2,p_3)$ around the real $\pi^{(0)}(0,\hat{p}_1,\hat{p}_2,\hat{p}_3)$, giving
\begin{gather}
-it + \sqrt{1- D^{\dagger}D/4} = 0 = 
-it + a(\pi^{(0)}(t,\hat{p}_1,\hat{p}_2,\hat{p}_3) - \pi^{(0)}(0,\hat{p}_1,\hat{p}_2,\hat{p}_3))\left.\frac{\partial\sqrt{1-D^{\dagger}D/4}}{\partial \hat{p}_0}\right|_{\pi^{(0)}}, \label{eq:56751} 
\end{gather}
thus, at small $t$, 
\begin{gather}
\Im(\arctan (a \pi^{(0)}/2) \sim \frac{\alpha t}{2}\nonumber\\
\alpha = \frac{1}{(1+(a\pi^{(0)}(0,\hat{p}_1,\hat{p}_2,\hat{p}_3)/2)^2)\frac{\partial\sqrt{1-D^{\dagger}D/4}}{\partial \hat{p}_0}}.
\end{gather}
At large enough $x$, the integrand will be exponentially suppressed for all $t$ except $tx < O(1)$, giving an integral of the form
\begin{gather}
\gamma_R^{(0)} \sim \int_0^{\infty} dt (t + \ldots) e^{-\alpha t x/a + 2ix/a(\arctan (a \pi^{(0)}(0,\hat{p}_1,\hat{p}_2,\hat{p}_3)/2) + \ldots},
\end{gather}
which will contain terms of the order of $a^2/x^2$. Thus, if there is a real solution to $1-D^{\dagger}D/4 = 0$, $\gamma_R^{(0)}$ will be non-local. If $\partial\sqrt{(1-D^{\dagger}D/4}/\partial \hat{p}_0 = 0$ at $\hat{p}_0 = \pi^{(0)}$ it would not affect the conclusions of this argument: one would simply have to incorporate higher order terms into the expansion of equation (\ref{eq:56751}), and be left with an integration such as $\int dt t e^{-x \sqrt{t}}$, which again has a power law decay in $x$.
\item Eigenvalues of O($2$): here there is a simple pole in $1/t$, and the integral over $t$ can be performed simply using contour integration. The pole at $t=0$ corresponds to a $\pi^{(0)}$ with large real part and small imaginary part. Therefore $\Im [\arctan (a \pi^{(0)}/2] \sim 0$ and for these eigenvectors $\gamma_R^{(\eta)}$ does not even have a polynomial decay in $x$.
\end{itemize} 
$\gamma_R^{(\eta)}$ for other even integer $\eta$ will also be non-local.

Therefore $\gamma_R^{(\eta)}$ is local when $\eta$ is an odd integer, and non-local if it is an even integer.\footnote{This corrects a statement I made in~\cite{Cundy:2009ab}.} Similar arguments can be developed to show that $\gamma_R^{(\eta)}$ is non-local for non-integer $\eta$.

It can be argued that this is enough for chiral symmetry on the lattice: the odd-integer $\eta$ $\gamma_5$-operators themselves form a closed group and are enough to define a chiral symmetry. However, the lattice $\mathcal{CP}$-transformation contains $\gamma_R^{(-\eta-1)}$, which is even integer and thus non-local. Does this matter? 
The continuum $\mathcal{CP}$ symmetry of the fermion fields is itself non-local, and we have no reason to desire that the fermion fields $\psi$ and $\overline{\psi}$ should be local (whatever that means in the case of the fields). No action or observable contains $\gamma_R^{(-\eta-1)}$. In these circumstances, there is no obvious reason to desire a local generator of lattice $\mathcal{CP}$ except to ensure that the lattice symmetry has a smooth limit to the continuum symmetry. And here we have a major conceptual problem: the difficulties on the lattice are concerned with the doublers of the zero modes. There is no continuum counterpart for these eigenvectors of the lattice overlap operator, and for physical momenta there is no non-analyticity in the Fourier Transform of the $\mathcal{CP}$ operators: the non analyticity occurs only at momentum of the order of the cut-off (or at infinite $\hat{p}$), and is contained within the ultra-violet physics. The difficulty is that we are trying to map a lattice theory with doublers (albeit with a mass of the order of the cut-off) to a continuum theory without doublers. One possible solution is to add additional fermion fields to the continuum action to simulate the lattice doublers, as is explored in~\cite{Cundy:newpaper}. If this is done, the portion of the overlap operator which is mapped to the physical modes has a local $\mathcal{CP}$ symmetry, and the portion which is mapped to the doublers a non-local $\mathcal{CP}$. It seems as though a non-local lattice $\mathcal{CP}$ symmetry transformation is unavoidable, but the symmetry is expressed in such a way that this non-local operator will not affect any observable.

\section{Conserved Currents and Ward identities}\label{sec:4}
%OK, here reproduce Mandula's paper and consider the generators of the symmetry, and the currents connected with the symmetry.
Including a matrix $\lambda$ which mixes the flavours, the change in the fermion fields under an infinitesimal local chiral rotation are
\begin{align}
\delta^{(\eta)}\psi =& i \epsilon(x)\gamma_R^{(\eta)}\lambda\psi\nonumber\\
\delta^{(\eta)}\overline{\psi} =& i \overline{\psi}\gamma_L^{(\eta)}\lambda\epsilon(x),\label{eq:cc1}
\end{align}
where $\epsilon(x)$ is diagonal in the spinor and colour indices.
The change in the action is
\begin{gather}
\Delta S = i \overline{\psi}\lambda(\gamma_L^{(\eta)}\epsilon D + D \epsilon\gamma_R^{(\eta)})\psi.
\end{gather}
The conserved current can be constructed by Noether's procedure. We write the action as
\begin{align}
\Delta S =& i \overline{\psi}\lambda(\epsilon(\gamma_L^{(\eta)} D + D \gamma_R^{(\eta)}) + [D,\epsilon]\gamma_R^{(\eta)} + [\gamma_L^{(\eta)},\epsilon] D)\psi\nonumber\\
=& i \overline{\psi}\lambda([D,\epsilon]\gamma_R^{(\eta)} + [\gamma_L^{(\eta)},\epsilon] D)\psi\nonumber\\
=& \sum_{\mu,x}\partial_{\mu}\epsilon(x) J^{(\eta)}_{\mu}(x),\label{eq:current1}
\end{align}
where $\partial_{\mu}$ is the forward normal derivative. We can write that
\begin{align}
[K,\epsilon] %=& \partial_{\mu} \epsilon(x) \gamma_5\left[(1-\gamma_{\mu})U_{\mu}(x)- (1+\gamma_{\mu})U_{\mu}^{\dagger}(x)\right]\nonumber\\
 =& \partial_{\mu}\epsilon(x) T_{\mu}(x)\\
[K^n,\epsilon] = & \sum_m K^m \partial_{\mu}\epsilon T_{\mu} K^{n-m-1}
\end{align}
where 
\begin{align}
T_{\mu}(x) =& \frac{\partial}{\partial U_{\mu}(x)} {K} U_{\mu}(x) - U^{\dagger}_{\mu}(x) \frac{\partial}{\partial U_{\mu}(x)} {K}\nonumber\\
=& \partial_{\mu} \epsilon(x) \gamma_5\left[(1-\gamma_{\mu})U_{\mu}(x)- (1+\gamma_{\mu})U_{\mu}^{\dagger}(x)\right]
.
\end{align}
 Therefore,
\begin{align}
[D,\epsilon(x)] =& \gamma_5 \mathfrak{J}= \gamma_5 \frac{1}{\pi} \int_{-\infty}^{\infty} dt \frac{1}{t^2 + K^2} (t^2 \partial_{\mu}(\epsilon) T_{\mu} - K\partial_{\mu}(\epsilon) T_{\mu} K) \frac{1}{t^2 + K^2}\label{eq:current2}\\
[\gamma_R^{(\eta)},\epsilon(x)] = & \gamma_5 \mathfrak{J}_C^{(\eta)} + \mathfrak{J}_{\epsilon} \sin\left((\eta + 1)\left(\frac{\pi}{2} - \theta\right)\right) + \sign(\gamma_5(D^{\dagger} - D)) \mathfrak{J}_S^{(\eta)}, \label{eq:current3}
\end{align}
where, for odd integer $\eta$ 
\begin{align}
\mathfrak{J}_{\epsilon} = &\frac{1}{\pi}\int_{-\infty}^{\infty} dt' \frac{1}{(t')^2 + ([D,\gamma_5])^2} ((t')^2[[D,\gamma_5],\epsilon] - [D,\gamma_5][[D,\gamma_5],\epsilon][D,\gamma_5]\frac{1}{(t')^2 + ([D,\gamma_5])^2}\nonumber\\
\mathfrak{J}_{C}^{(\eta)}=&-\sum_{m=0}^{(\eta+1)/2}\;\;\sum_{k=0}^{m-1} (-1)^m \frac{(\eta+1)!}{(2m)! (\eta + 1 - 2m)!}\left(1-\frac{D^{\dagger}D}{4}\right)^k \{\gamma_5 D,\gamma_5[D,\epsilon]\}\nonumber\\
&\phantom{space}\left(1-\frac{D^{\dagger}D}{4}\right)^{m-1-k} \left(\frac{D^{\dagger}D}{4}\right)^{(\eta + 1)/2 - m} +\nonumber\\
&\sum_{m=0}^{(\eta+1)/2}\;\;\sum_{k=0}^{(\eta + 1)/2 - m-1} (-1)^m \frac{(\eta+1)!}{(2m)! (\eta + 1 - 2m)!}\left(1-\frac{D^{\dagger}D}{4}\right)^m  \left(\frac{D^{\dagger}D}{4}\right)^{k}  \nonumber\\
&\phantom{space}\{\gamma_5 D,\gamma_5[D,\epsilon]\}\left(\frac{D^{\dagger}D}{4}\right)^{(\eta - 1)/2 - m-k},
\end{align}
and $\mathfrak{J}_{S}^{(\eta)}$ can be found in the same way from the differential of
\begin{gather}
-\sum_{m=0}^{(\eta-1)/2} (-1)^m \frac{(\eta+1)!}{(2m+1)! (\eta+1-2m-1)!} \sqrt{1-\frac{D^{\dagger}D}{4}}\sqrt{D^{\dagger}D}\left(1-\frac{D^{\dagger}D}{4}\right)^m \left( \frac{D^{\dagger} D}{4}\right)^{(\eta-1)/2 - m}\nonumber
\end{gather}
%I really can't be bothered to type this out ...
To derive these expressions, $\cos(\eta+1)(\pi/2-\theta)$ and $\sin (\eta+1)(\pi/2-\theta)$ were expanded using equations (\ref{eq:cosntheta}) and (\ref{eq:sinntheta}), and the definition of $\theta$ was used from equation (\ref{eq:gamma_eta0}). The conserved current can be constructed from equations (\ref{eq:current1}), (\ref{eq:current2}) and (\ref{eq:current3}).

The Ward identity is derived by considering the expectation value of an operator $\mathcal{O}$. We can write that
\begin{gather}
\langle\mathcal{O}[\psi,\overline{\psi}]\rangle = \int d\psi d\overline{\psi} d U \mathcal{O}[\psi,\overline{\psi}] e^{-S_g[U] - S_f[\psi,\overline{\psi},U]}.
\end{gather}
When the change (\ref{eq:cc1}) in the fermion fields is applied, the expectation value is invariant, which means that the contributions at $O(\epsilon(x))$ must cancel. There are three places within the expectation value where these effects can enter the expression: from the change in the action, the operator, and the integration measure. The change from the action leads to a term 
\begin{align}
-\int d\psi d\overline{\psi} d U \mathcal{O}[\psi,\overline{\psi}] e^{-S_g[U] - S_f[\psi,\overline{\psi},U]} \partial_{\mu}(\epsilon) J^{(\eta)}_{\mu} = \epsilon(x) \langle \mathcal{O} \partial_{\mu}^* J^{(\eta)}_{\mu}(x)\rangle,
\end{align} 
where $\partial_{\mu}^*$ is the backwards non-covariant derivative and I have used the identity 
\begin{gather}\sum_x (\partial_{\mu}(\epsilon(x)) J^{(\eta)}_{\mu}(x) = -\sum_x \epsilon(x) \partial_{\mu}^* J^{(\eta)}(x).\nonumber
\end{gather}
 From the change in the operator, we obtain
\begin{gather}
\left\langle \frac{\partial \mathcal{O}}{\partial \psi(x)} (\gamma_R^{(\eta)} \psi)(x) + (\overline{\psi} \gamma^{(\eta)}_L)(x)  \frac{\partial \mathcal{O}}{\partial \overline{\psi}(x)}\right\rangle = \langle \delta^{(\eta)} \mathcal{O}[\psi,\overline{\psi}] \rangle. 
\end{gather} 
Finally, from the change in the measure, we obtain obtain the topological charge
\begin{gather}
\langle \Tr (\gamma_L^{(\eta)} + \gamma_R^{(\eta)}) \mathcal{O}\rangle = - \langle Q[U]\mathcal{O}\rangle.
\end{gather}
If the expectation values are taken in a fixed topological sector, then the Ward identity is
\begin{gather}
\langle \delta^{(\eta)} \mathcal{O}\rangle_Q - Q \langle \mathcal{O} \rangle_Q + \partial_{\mu}^*\langle  J^{(\eta)}(x) \mathcal{O}\rangle_Q.
\end{gather}
Demonstrating the locality (or otherwise) of the conserved current for general $\eta$ is clearly a non-trivial exercise, and beyond the intended scope of this work. When $\eta = \pm 1$, and the chiral symmetry transformations reduce to their canonical form, the expression for the current simplifies to the well known formula~\cite{Kikukawa:1998py}. For other $\eta$, the locality can be considered using a method similar to that of the previous section. Once again, the currents are local for odd integer $\eta$ and non-local otherwise. 

\section{Weyl Fermions on the lattice}\label{sec:7}

The construction of a Weyl fermion action is now straightforward. Using the Ginsparg-Wilson relation, we can write 
\begin{gather}
S^{(\eta)} = \overline{\psi} D \psi = \frac{1}{4}\overline{\psi}(1+\gamma^{(\eta)}_L)D(1-\gamma^{(\eta)}_R)\psi + \frac{1}{4}\overline{\psi}(1-\gamma^{(\eta)}_L)D(1+\gamma^{(\eta)}_R)\psi. 
\end{gather}
As already discussed, this chiral Lagrangian is $\mathcal{CP}$ invariant. We can write a new action in terms of Weyl fermions $\psi^{(\eta)}_+$ and $\psi^{(\eta)}_-$,
\begin{gather}
S^{(\eta)} = \overline{\psi}_+^{(\eta)} D \psi_-^{(\eta)} + \overline{\psi}_-^{(\eta)}D\psi_+^{(\eta)},
\end{gather}
where
\begin{align}
\psi_{\pm}^{(\eta)} =& \frac{1}{2}(1\pm \gamma^{(\eta)}_R) \psi;& \overline{\psi}^{(\eta)}_{\pm} =& \overline{\psi}\frac{1}{2}(1\pm \gamma^{(\eta)}_L).
\end{align}
It is not immediately clear that the measure of this transformation of the fermion variables is gauge invariant, because $\gamma_L^{(\eta)}$ and $\gamma_R^{(\eta)}$ depend on the gauge fields. If the zero modes are in the positive chiral sector (the case where they are in the negative chiral sector is analogous), then the measure can be calculated in terms of the basis used in section \ref{sec:4.1}
\begin{align}
\psi_+^{(\eta)} =& c_i^{(\eta)} \phi^{(\eta)}_{+i} + c_0 \phi_0; &\overline{\psi}_-^{(\eta)} =& \overline{c}_i\phi_{-i}^{(\eta)}\hat{\gamma}_R^{(-\eta-1)} + \overline{c}_0\phi_0 = \overline{c}_i\phi_{-i}^{(-\eta)}\gamma_5 + \overline{c}_0\phi_0.
\end{align}
The Jacobian for the change in the measure for an infinitesimal change in the gauge field is $ e^{-i\mathfrak{L}_{\xi}^{(\eta)}}$, where $\mathfrak{L}_{\xi}^{(\eta)}$ is given by equation (\ref{eq:Ljac}):
\begin{align}
\mathfrak{L}_{\xi}^{(\eta)} =& i\sum_i \left[\braket{\phi^{(\eta)}_{+i}}{\delta_{\xi} \phi^{(\eta)}_{+i}} +\braket{\delta_{\xi}\phi^{(-\eta)}_{-i}}{ \phi^{(-\eta)}_{-i}}\right] + \braket{\delta_{\xi} \phi_0}{\phi_0} + \braket{\phi_0}{\delta_{\xi} \phi_0},\label{eq:Ljac2}
\end{align} 
For the zero modes (or, equivalently their doublers if the Weyl fermion is in the chiral sector containing the doublers rather than the zero modes),
\begin{gather}
\braket{\delta_{\xi} \phi_0}{\phi_0} + \braket{\phi_0}{\delta_{\xi} \phi_0} = \delta_{\xi}\braket{\phi_0}{\phi_0} = 0.
\end{gather}
For the non-zero modes, I use equation (\ref{eq:phialpha}) to write,
\begin{gather}
(\phi^{(\eta)}_{+i},\delta_{\xi} \phi^{(\eta)}_{+i}) = \frac{\sin 2\alpha^{(\eta)}}{2}((H_{-i},\delta_{\xi} H_{+i}) + (H_{+i},\delta_{\xi} H_{-i})).  
\end{gather}
The differential of the  eigenvectors is (for example, see~\cite{Cundy:2007df}),
\begin{gather}
\delta_{\xi} \ket{H_{+i}} = (1-\ket{H_{+i}}\bra{H_{+i}}) \frac{1}{H_{+i} - \lambda_i} \delta_{\xi}H \ket{H_{+i}},
\end{gather}
which gives,
\begin{gather}
\bra{H_{-i}}\delta_{\xi} \ket{H_{+i}} + \bra{H_{+i}}\delta_{\xi} \ket{H_{-i}} = \frac{1}{2\lambda_i}(\bra{H_{-i}} \delta_{\xi} H \ket{H_{+i}} - \bra{H_{+i}} \delta_{\xi} H \ket{H_{-i}}).\label{eq:non2}
\end{gather}
For a gauge transformation, with the gauge field is in a representation $R$, 
\begin{gather}
A_{\mu} \rightarrow A_{\mu} - \partial_{\mu} \xi,
\end{gather}
the change in the Dirac operator is~\cite{Luscher:1999un}
\begin{gather}
\delta H = \gamma_5 \delta D = \gamma_5 [R(\xi),D].
\end{gather}
Inserting this relation into equation (\ref{eq:non2}) and using the relation 
\begin{gather}
 \sum_i \frac{1}{\lambda}(\ket{H_{+i}}\bra{H_{-i}} -\ket{H_{-i}}\bra{H_{+i}} )= \frac{1}{D^{\dagger}D\sqrt{1-\frac{D^{\dagger}D}{4}}} (D^{\dagger} - D)(1-\ket{\phi_0}\bra{\phi_0} - \ket{\phi_2}\bra{\phi_2})
\end{gather}
gives
\begin{gather}
(\phi^{(\eta)}_{+},\delta_{\xi} \phi^{(\eta)}_{+}) = -\frac{1}{2}\Tr_{\text{Non-zero eigenvectors}}\left[\gamma_5 R(\xi)\sqrt{1-\frac{D^{\dagger}D}{4}} \sin (2 \alpha^{(\eta)})\right].
\end{gather}
Since $R(\xi)$, $\alpha^{(\eta)}$ and $D^{\dagger}D$ commute with $\gamma_5$ this trace is zero. By repeating this argument for the remaining term in equation (\ref{eq:Ljac2}), it can be shown that $\mathfrak{L}_{\xi}^{(\eta)}=0$ and therefore the measure of the Weyl fermion is invariant under gauge transformations.

\section{Majorana Fermions on the lattice}\label{sec:8}
Majorana fermions are neutral fermions which are their own anti-particle. Although no such fermions have been found in practice, they play a role in various beyond the standard model scenarios, such as, when coupled with a Yukawa coupling, massive neutrinos and certain super-symmetric models. 
\subsection{Majorana Fermions in the continuum}
$C$ is the charge conjugation operator, which has the properties
\begin{align}
C^{-1}\gamma_5^TC =& \gamma_5\nonumber\\
C^{-1}D^T[U]C = & D[U_C] \nonumber\\
C^{\dagger}C =& 1\nonumber\\
C\gamma_{\mu} C^{-1} = - \gamma_{\mu}^T.
\end{align}
Because they are neutral particles, Majorana fermions satisfy the constrain $\psi^* = \mathfrak{B}\psi$, where $\mathfrak{B} = \gamma_5 C$. To fulfil both this condition and $(\psi^*)^* = \psi$, it is necessary to double the fermionic degrees of freedom, so that 
\begin{gather}
\psi = \left(\begin{array}{c} \psi_1\\ \psi_2\end{array}\right)
\end{gather}
we can define a charge conjugation operator in this enlarged space as
\begin{align}
\mathcal{C} =& \left(\begin{array}{c c}
0&C\\
C&0
\end{array}\right)\nonumber \\
\mathcal{B} =&\left(\begin{array}{c c}
0&\mathfrak{B}\\
-\mathfrak{B}&0
\end{array}\right)
\end{align}
with a corresponding Dirac operator $\mathcal{D}$,
\begin{align}
\mathcal{D} =&\left(\begin{array}{c c}
0&D\\
D&0
\end{array}\right).
\end{align}
This Dirac operator satisfies
\begin{align}
(\mathcal{CD})^T =& -\mathcal{CD};&(CD)^T =& -C D;& D^* =& \mathfrak{B} D \mathfrak{B}^{-1};\label{eq:CDisantisymmetric}
\end{align}
The fermionic action is given as
\begin{gather}
S_f = \int d^4x \psi^T \mathcal{CD}\psi = \int d^4x (\psi_1^T CD \psi_2 + \psi_2^T CD \psi_1).
\end{gather}
The action also contains a chirally symmetric Yukawa term, coupling the fermion field to a scalar field $S$,
\begin{gather}
S_Y[\psi,S] = g \int d^4 x \left(\psi^T \mathcal{CP}S^{\dagger}\psi + \psi^T \mathcal{C}(1-\mathcal{P})S\psi\right), 
\end{gather}
where $S$ takes the form,
\begin{gather}
S = \left(\begin{array}{l l} 
0&\upsilon\\
\upsilon & 0
\end{array}
\right),
\end{gather}
for a complex scalar field $\upsilon$,
and $\mathcal{P}$ is a projection operator,
\begin{gather}
\mathcal{P} = \frac{1}{2}(1+\Gamma_5).
\end{gather}
$\Gamma_5$ is the higher dimensional equivalent of $\gamma_5$,
\begin{align}
\Gamma_5=&\left(\begin{array}{c c}
\gamma_5&0\\
0&-\gamma_5
\end{array}\right).
\end{align}
Under an infinitesimal chiral symmetry transformation,
\begin{align}
\psi \rightarrow& (1+i\epsilon \Gamma_5)\psi; & S\rightarrow (1-2i\epsilon) S,
\end{align}
the action $S=S_f + S_Y$ remains invariant if
\begin{gather}
\Gamma_5 \mathcal{D} - \mathcal{D}\Gamma_5 = 0,
\end{gather}
which is satisfied for,
\begin{gather}
D\gamma_5 + \gamma_5 D = 0.
\end{gather}
\subsection{Lattice representation}
Suppose that we have some Ginsparg-Wilson relation,
\begin{gather}
\gamma_L^{(\eta)} D + D \gamma_R^{(\eta)} = 0,
\end{gather}
where $\gamma_L^{(\eta)}$ and $\gamma_R^{(\eta)}$ both reduce to $\gamma_5$ in the continuum limit. It is then natural to construct projection operators,
\begin{align}
\Gamma_{L,R}^{(\eta)} =& \left(\begin{array}{l l} \gamma_{L,R}^{(\eta)}&0\\
0&-\gamma_{L,R}^{(\eta)}\end{array}\right)\nonumber\\
P_{L,R} =& \frac{1}{2}\left(1+\Gamma_{L,R}^{(\eta)}\right).
\end{align}
In the higher dimensional representation, the Ginsparg-Wilson equation becomes.
\begin{gather}
\Gamma_L^{(\eta)} D - D \Gamma_R^{(\eta)} = 0.
\end{gather}
From the continuum action, we may apply the fermion blockings and use the relations
\begin{align}
(B^{(\eta)})^T C =& C \overline{B}^{(-\eta)}\nonumber\\
\intertext{and}
\overline{B}^{(-\eta)}(\overline{B}^{(\eta)})^{-1} = \gamma_R^{(-1-\eta)}\gamma_5 \nonumber
\end{align} 
to write down a lattice action
\begin{align}
S_f =& \int d^4 x \psi^T\mathcal{C}^{(\eta)}\mathcal{D} \psi\nonumber\\
S_Y = &\int d^4 x g \int d^4 x \left(\psi^T \mathcal{C}^{(\eta)}\mathcal{P}^{(\eta)}_LS^{\dagger}(1-\mathcal{P}^{(\eta)}_R)\psi + \psi^T \mathcal{C}^{(\eta)}(1-\mathcal{P}^{(\eta)}_L)S\mathcal{P}^{(\eta)}_R\psi\right),
\end{align}
where
\begin{gather}
C^{(\eta)} = C\gamma_R^{(-\eta - 1)} \gamma_5.
\end{gather}
The charge conjugation operator acts on $\gamma_R^{(\eta)}$ according to
\begin{align}
C^{-1}\gamma_R^{(\eta)} C =& \gamma_5^T \cos((1+\eta)(\pi/2 - \theta)) + \sign((\gamma_5)^T (D^{\dagger}-D)^T)\sin((1+\eta)(\pi/2 - \theta))\nonumber\\
=&\gamma_5^T\cos((1+\eta)(\pi/2 - \theta)) - (\sign(\gamma_5(D^{\dagger} - D)))^T \sin((1+\eta)(\pi/2 - \theta))\nonumber\\
=&(\gamma_R^{(-\eta - 2)})^T.
\end{align}
Therefore, using equation (\ref{eq:gamma_multiply}),
\begin{align}
\gamma_5\gamma_R^{(-\eta - 1)}C^{-1}(\gamma_R^{(\eta)})^TC \gamma_R^{(-\eta - 1)}\gamma_5 = \gamma_L^{(\eta)},
\end{align}
giving
\begin{gather}
(\mathcal{C}^{(\eta)})^{-1}(\Gamma_R^{(\eta)})^T\mathcal{C}^{(\eta)} = -\Gamma_L^{(\eta)}.\label{eq:upthere}
\end{gather}
An infinitesimal chiral symmetry rotation is
\begin{align}
\psi \rightarrow& (1+i\epsilon\Gamma_R) \psi; &\psi^T \rightarrow& \psi^T(1+i\epsilon\Gamma_R^T); & S \rightarrow&(1+2i\epsilon) S.
\end{align}
The fermion action transforms as
\begin{gather}
S_f\rightarrow S_f + i \epsilon \int d^4 x\psi^T \mathcal{C} ( \mathcal{C}^{-1} \Gamma_R^T \mathcal{C} \mathcal{D} + \mathcal{D} \gamma_R) \psi,
\end{gather}
which, given equation (\ref{eq:upthere}) is symmetric under chiral symmetry. 
The change in the Yukawa action under chiral symmetry is
\begin{align}
&i\epsilon(\psi^T(\Gamma_R^{(\eta)})^T\mathcal{C}^{(\eta)}\mathcal{P}^{(\eta)}_LS^{\dagger} (1-\mathcal{P}^{(\eta)}_R) \psi + 2\psi^T\mathcal{C}^{(\eta)}\mathcal{P}^{(\eta)}_LS^{\dagger} (1-\mathcal{P}^{(\eta)}_R) \psi )+\phantom{a}\nonumber\\ &i\epsilon(\psi^T\mathcal{C}^{(\eta)}\mathcal{P}^{(\eta)}_LS^{\dagger} (1-\mathcal{P}^{(\eta)}_R) \Gamma_R^{(\eta)} \psi+\psi^T(\Gamma_R^{(\eta)})^T\mathcal{C}^{(\eta)}(1-\mathcal{P}^{(\eta)}_L)S (\mathcal{P}^{(\eta)}_R) \psi) -\phantom{a}\nonumber\\
&i\epsilon( 2\psi^T\mathcal{C}^{(\eta)}(1-\mathcal{P}_L^{(\eta)})S \mathcal{P}_R^{(\eta)} \psi  -\psi^T\mathcal{C}^{(\eta)}(1-\mathcal{P}_L^{(\eta)})S \mathcal{P}_R^{(\eta)} \Gamma_R^{(\eta)} \psi) =0  
\end{align}
since $\mathcal{P}^{(\eta)}_{L,R}\Gamma^{(\eta)}_{L,R} = \mathcal{P}^{(\eta)}$ and $(1-\mathcal{P}^{(\eta)}_{L,R})\Gamma_{L,R}^{(\eta)} = - (1-\mathcal{P}_{L,R}^{(\eta)})$. Therefore this term in the action also remains symmetric under chiral symmetry. Thus this action correctly transforms under chiral symmetry.

%Not happy with the way I'm adding this mass - it doesn't feel natural. But then, I guess that I have to worry about masses in general. */
The massive overlap operator can be written as $D[m] = (1-\frac{m}{2m_W} D[0]) + \frac{m}{m_W}$. The massive Majorana fermion action is given by $S_f[m] = \sqrt{1-(m/2m_W)^2}S_f[0] + S_M$, where the form of $S_M$ can be derived from the continuum mass using the blockings,
\begin{align}
S_{M} =& i\frac{m}{m_W} \psi_0^T \mathcal{C} \left(\begin{array}{l l} 0 &1\\1&0\end{array}\right) \Gamma_5 \psi_0\nonumber\\
=&\psi^T (B^{(\eta)})^T \mathcal{C} \left(\begin{array}{l l} 0 &1\\1&0\end{array}\right) \Gamma_5  B^{(\eta)} \psi\nonumber\\
=&\psi^T \mathcal{C}  \left(\begin{array}{l l} 0 &1\\1&0\end{array}\right) \overline{B}^{(-\eta)} \Gamma_5 B^{(\eta)} \psi\nonumber\\
=&\psi^T \mathcal{C}  \left(\begin{array}{l l} 0 &1\\1&0\end{array}\right) \Gamma_R^{(-\eta - 1)}\psi.
\end{align}
The partition function for the fermion fields can be written as
\begin{align}
Z= \int d\psi_1 d\psi_2 &e^{-\psi_1^T(C\gamma_R^{-1-\eta}\gamma_5 D[0]\sqrt{1-(m/2m_W2} + i (m/m_W) C \gamma_R^{-1-\eta})\psi_1}\nonumber\\
&\phantom{spacespace}e^{- \psi_2^T(C\gamma_R^{-1-\eta}\gamma_5D[0]\sqrt{1-(m/2m_W)^2} - i (m/m_W) C \gamma_R^{-1-\eta}) \psi_2},  
\end{align}
using a standard result for Grassman variables $a$ and an anti-symmetric matrix $M$,
\begin{gather}
\int da e^{-a^T Ma} = Pf(M),
\end{gather}
where $Pf(M)$ is the Pfaffian of the matrix, by noting equation (\ref{eq:CDisantisymmetric}), we have
\begin{align}
Z =& Pf(C\gamma_R^{(-1-\eta)}\gamma_5D[0]\sqrt{1-(m/2m_W)^2} + i (m/m_W) C \gamma_R^{(-1-\eta)})\nonumber\\
&\phantom{spacespace}Pf(C\gamma_R^{(-1-\eta)}\gamma_5D[0]\sqrt{1-(m/2m_W)^2} - i (m/m_W) C \gamma_R^{(-1-\eta)}) \nonumber\\
=& Pf( \gamma_5 D[0] \gamma_5 D[0] (1-(m/2m_W)^2) + (m/m_W)^2)\nonumber\\
=& Pf( \gamma_5 D[0] (1-m/(2m_W)) + \gamma_5(m/m_W))^2 \nonumber\\
=& \det[D[m]].
\end{align}
Thus this Majorana fermion representation is both Pfaffian and, upon integration over the fermion variables, gives the correct partition function.

\section{Propagators in the presence of a Higgs field}\label{sec:8b}
In the Glashow-Weinberg-Salam theory of the weak interaction, the Higgs boson, $S_0$, interacts with left and right handed fermion fields by a Yukawa coupling
\begin{gather}
\Delta \mathcal{L} = - \frac{\mathcal{\lambda}}{4}\big[ \overline{\psi}_0(1-\gamma_5) S_0 (1-\gamma_5)\psi_0 +  \overline{\psi}_0(1+\gamma_5) S^{\dagger}_0 (1+\gamma_5)\psi_0\big],
\end{gather}
where $\mathcal{\lambda}$ is a coupling constant and $\psi_0$ and $\overline{\psi}_0$ are the continuum fermion fields. The lattice fermion fields are constructed from
\begin{align}
\psi_0 =& B \psi, & \overline{\psi}_0 =& \overline{\psi}\; \overline{B}.
\end{align}
The lattice action is therefore
\begin{align}
\Delta \mathcal{L} = - \frac{\mathcal{\lambda}}{4}(\overline{\psi}(1-\gamma_L^{(\eta)})S(1-\gamma_R^{(\eta)})\psi + \overline{\psi}(1+\gamma_L^{(\eta)})S^{d}(1+\gamma_R^{(\eta)})\psi),
\end{align} 
where the lattice Higgs field is
\begin{align}
S =& (\overline{B}^{(\eta)})^{-1}S_0 (B^{(\eta)})^{-1}, & S^d = &(\overline{B}^{(\eta)})^{-1}S_0^{\dagger} (B^{(\eta)})^{-1}. 
\end{align}
The change of the measure for the transformation of the Higgs field can be used to cancel the Jacobian for the fermion fields, or, as in the original proposal for the renormalisation group construction of the lattice theory, absorbed into the blocking of the gauge fields. The two Higgs fields are no longer Hermitian conjugates unless the blockings are symmetric ($\eta = 0$), but must be treated as independent variables. Under $\mathcal{CP}$,
\begin{align}
\mathcal{CP}:S =& (\gamma_R^{(-1-\eta)}\gamma_5S\gamma_5\gamma_R^{(-1-\eta)})^T\\
\mathcal{CP}:S^d =& (\gamma_R^{(-1-\eta)}\gamma_5S^d\gamma_5\gamma_R^{(-1-\eta)})^T.
\end{align}
These relations can be proved following the methods already considered for the fermion fields.

Replacing the Higgs field $S$ with its vacuum expectation value $\upsilon$ generates a mass for the fermion fields, and models using spontaneous symmetry breaking of a scalar field are the only known ways to generate the quark mass in the standard electroweak theory. To implement this interaction on the lattice requires well defined chiral projectors, which requires Ginsparg-Wilson fermions, and simulations using the standard projectors have been carried out~\cite{Luscher:1998pqa,Gerhold:2010bh,Gerhold:2009ub,Gerhold:2007yb}. However, it is clear that if the right and left handed vectors are not related to each other by $\mathcal{CP}$ symmetry, then, if the Higgs has a non-zero expectation value, there will be a change in the action under $\mathcal{CP}$ symmetry and therefore a change in the quark propagators. In~\cite{Fujikawa:2002vj} this change was calculated and found to be non-local, suggesting that it may continue to the continuum limit (although the authors of that expected that the difference would nonetheless vanish in the continuum, given that the Lagrangian itself has the correct continuum limit). It is therefore necessary to consider the implications of $\mathcal{CP}$ violation on the propagators in the presence of the Higgs field.

Consider the action,
\begin{gather}
S = 
\left(\begin{array}{l l}
\overline{\psi}_- &\overline{\psi}_+
\end{array}\right)
\left(\begin{array}{l l}
D_-^{(\eta)}&\upsilon^{(\eta)}_{-+}\\
\upsilon^{(\eta)}_{+-}&D_+^{(\eta)}
\end{array}\right)
\left(\begin{array}{l}
\psi_-\\
\psi_+
\end{array}\right)
 +
 \left(\begin{array}{l l}
\overline{\psi}_- &\overline{\psi}_+
\end{array}\right)
\left(\begin{array}{l}
\chi_-\\
\chi_+
\end{array}\right)
+
\left(\begin{array}{l l}
\overline{\chi}_- &\overline{\chi}_+
\end{array}\right)
\left(\begin{array}{l}
\psi_-\\
\psi_+
\end{array}\right)
\end{gather} 
where
\begin{align}
D_-^{(\eta)} =& \frac{1}{4}(1-\gamma_L^{(\eta)})D(1+\gamma_R^{(\eta)})\nonumber\\
D_+^{(\eta)} =& \frac{1}{4}(1+\gamma_L^{(\eta)})D(1-\gamma_R^{(\eta)})\nonumber\\
\upsilon_{-+}^{(\eta)} =& \frac{\mathcal{\lambda}}{4}(1-\gamma_L^{(\eta)})S(1-\gamma_R^{(\eta)})\nonumber\\
\upsilon_{+-}^{(\eta)} =& \frac{\mathcal{\lambda}}{4}(1+\gamma_L^{(\eta)})S^{d}(1+\gamma_R^{(\eta)})
\end{align}
and $S$ is a scalar field. We can use a Schur decomposition to transform to a new basis
\begin{align}
\left(\begin{array}{l l}
\overline{\psi}_- &\overline{\psi}_+
\end{array}\right) =&
\left(\begin{array}{l l}
\overline{\psi}'_- &\overline{\psi}'_+
\end{array}\right)\left(\begin{array}{l l}
1&0\\
-\upsilon^{(\eta)}_{+-} (D_-^{(\eta)})^{-1}&1
\end{array}\right)\nonumber\\
\left(\begin{array}{l}
\psi_-\\
\psi_+
\end{array}\right)
=&\left(\begin{array}{l l}
1&-(D_-^{(\eta)})^{-1}\upsilon^{(\eta)}_{-+}\\
0&1
\end{array}\right)
\left(\begin{array}{l}
\psi'_-\\
\psi'_+
\end{array}\right)
\nonumber\\
\left(\begin{array}{l l}
\overline{\chi}'_- &\overline{\chi}'_+
\end{array}\right) =&
\left(\begin{array}{l l}
\overline{\chi}'_- &\overline{\chi}_+
\end{array}\right)
\left(\begin{array}{l l}
1&-(D_-^{(\eta)})^{-1}\upsilon^{(\eta)}_{-+}\\
0&1
\end{array}\right)
\nonumber\\
\left(\begin{array}{l}
\chi'_-\\
\chi'_+
\end{array}\right)
=&\left(\begin{array}{l l}
1&0\\
-\upsilon^{(\eta)}_{+-} (D_-^{(\eta)})^{-1}&1
\end{array}\right)
\left(\begin{array}{l}
\chi_-\\
\chi_+
\end{array}\right).
\end{align}
In these new coordinates, the action is
\begin{align}
S = &
\left(\begin{array}{l l}
\overline{\psi}'_- &\overline{\psi}'_+
\end{array}\right)
\left(\begin{array}{l l}
D_-^{(\eta)}&0\\
0&D_+^{(\eta)} - \upsilon^{(\eta)}_{+-} (D_-^{(\eta)})^{-1} \upsilon^{(\eta)}_{-+}
\end{array}\right)
\left(\begin{array}{l}
\psi'_-\\
\psi'_+
\end{array}\right)+\nonumber\\
 &
 \left(\begin{array}{l l}
\overline{\psi}'_- &\overline{\psi}'_+
\end{array}\right)
\left(\begin{array}{l}
\chi'_-\\
\chi'_+
\end{array}\right)
+
\left(\begin{array}{l l}
\overline{\chi}'_- &\overline{\chi}'_+
\end{array}\right)
\left(\begin{array}{l}
\psi'_-\\
\psi'_+
\end{array}\right).
\end{align}
Following equation (\ref{eq:gf}), we can write the generating function $Z_F[\omega,U,\eta,\chi,\overline{\chi},\phi]$ as
\begin{align}
Z_F[\omega,U,\eta,\chi,\overline{\chi},\phi]=&e^{i\theta[\omega,U,\eta]}\det(D_-^{(\eta)}(1-\phi_0\phi_0^{\dagger})) \left(\overline{\chi}'_-\frac{1}{2}(1+\gamma_R^{(\eta)}),\phi_0\right)\left(\phi_0\frac{1}{2}(1-\gamma_L^{(\eta)}),\chi'_-\right)\nonumber\\
&\phantom{space}e^{\overline{\chi}_-'  {G'_-}^{(\eta)} \chi'_-}\det((D_+^{(\eta)} - \upsilon_{+-}^{(\eta)}D_-^{-1}\upsilon_{-+}^{(\eta)})(1-\phi_0\phi_0^{\dagger}))\nonumber\\
&\phantom{space} \left(\overline{\chi}'_+\frac{1}{2}(1-\gamma_R^{(\eta)}),\phi_0\right)\left(\phi_0\frac{1}{2}(1+\gamma_L^{(\eta)}),\chi'_+\right)e^{\overline{\chi}'_+ {G'_+}^{(\eta)} \chi_+}\nonumber\\
=&e^{i\theta[\omega,U,\eta]}\det(D_-^{(\eta)}(1-\phi_0\phi_0^{\dagger})) \left(\overline{\chi}_-\frac{1}{2}(1+\gamma_R^{(\eta)}),\phi_0\right)\left(\phi_0\frac{1}{2}(1-\gamma_L^{(\eta)}),\chi_-\right)\nonumber\\
&\phantom{space}\det((D_+^{(\eta)} - \upsilon_{+-}^{(\eta)}D_-^{-1}\upsilon_{-+}^{(\eta)})(1-\phi_0\phi_0^{\dagger}))\nonumber\\
&\phantom{space} \left((\overline{\chi}_+ - \overline{\chi}_- \left(D_-^{(\eta)}\right)^{-1}\upsilon_{-+}^{(\eta)})\frac{1}{2}(1-\gamma_R^{(\eta)}),\phi_0\right)\nonumber\\
&\phantom{space}\left(\phi_0\frac{1}{2}(1+\gamma_L^{(\eta)}),\chi_+-\upsilon^{(\eta)}_{+-}(D_-^{(\eta)})^{-1}\chi_-\right)\nonumber\\
&\phantom{space}e^{\overline{\chi}_- G_{--}^{(\eta)} \chi_- + \overline{\chi}_+ G_{+-}^{(\eta)} \chi_- +\overline{\chi}_- G_{-+}^{(\eta)} \chi_+ +\overline{\chi}_+ G_{++}^{(\eta)} \chi_+},
\end{align}
where
\begin{align}
{G'_-}^{(\eta)} =& \frac{1}{4}(1+ \gamma_R^{(\eta)})\frac{1}{D_-^{(\eta)}}(1-\gamma_L^{(\eta)})(1-\ket{\phi_0}\bra{\phi_0})\nonumber\\
{G'_+}^{(\eta)} =& \frac{1}{4}(1- \gamma_R^{(\eta)})\frac{1}{D_+^{(\eta)} - \upsilon^{(\eta)}_{+-} (D_-^{(\eta)})^{-1}\upsilon^{(\eta)}_{-+}}(1+\gamma_L^{(\eta)})(1-\ket{\phi_0}\bra{\phi_0})\nonumber\\
G_{--}^{(\eta)}[U,S,S^{d}] = &\frac{1}{4}(1+ \gamma_R^{(\eta)})\frac{1}{D_-^{(\eta)} - \upsilon^{(\eta)}_{-+} (D_+^{(\eta)})^{-1}\upsilon^{(\eta)}_{+-}} (1-\gamma_L^{(\eta)})(1-\ket{\phi_0}\bra{\phi_0})\nonumber\\
G_{++}^{(\eta)}[U,S,S^{d}] = &\frac{1}{4}(1- \gamma_R^{(\eta)})\frac{1}{D_+^{(\eta)} - \upsilon^{(\eta)}_{+-} (D_-^{(\eta)})^{-1}\upsilon^{(\eta)}_{-+}} (1+\gamma_L^{(\eta)})(1-\ket{\phi_0}\bra{\phi_0})\nonumber\\
G_{-+}^{(\eta)}[U,S,S^{d}] = &\frac{1}{4}(1+ \gamma_R^{(\eta)})\frac{1}{\upsilon^{(\eta)}_{+-} - D_+^{(\eta)}(\upsilon^{(\eta)}_{-+})^{-1} D_-^{(\eta)}}(1-\ket{\phi_0}\bra{\phi_0}) (1+\gamma_L^{(\eta)})\nonumber\\
G_{+-}^{(\eta)}[U,S,S^{d}] = &\frac{1}{4}(1- \gamma_R^{(\eta)})\frac{1}{\upsilon^{(\eta)}_{-+} - D_-^{(\eta)}(\upsilon^{(\eta)}_{+-})^{-1} D_+^{(\eta)}}(1-\ket{\phi_0}\bra{\phi_0}) (1-\gamma_L^{(\eta)})
\end{align}
Under $\mathcal{CP}$,
\begin{align}
\mathcal{CP}:G^{\eta}[U,S,S^d] \rightarrow& \gamma_5 \gamma_R^{-1-\eta} G^{(-\eta)}[U,\gamma_R^{(-1-\eta)}\gamma_5S\gamma_5\gamma_R^{(-1-\eta)},\gamma_R^{(-1-\eta)}\gamma_5S^{d}\gamma_5\gamma_R^{(-1-\eta)} \gamma_R^{-1-\eta}\gamma_5\nonumber\\
=&G^{\eta}[U,S,S^d].
\end{align}
Therefore the propagators are invariant under $\mathcal{CP}$ even in the presence of a non-vanishing Higgs field.
\section{Conclusions}\label{sec:9}
I have shown that the problems concerning $\mathcal{CP}$ invariance in Ginsparg-Wilson chiral gauge theories are a result of naively applying the continuum $\mathcal{CP}$ symmetry to the lattice. However, by considering the Ginsparg Wilson method, I have shown that this is incorrect, and a more natural definition of lattice $\mathcal{CP}$ flows from the same methodology from which lattice chiral symmetry was derived. Using this lattice $\mathcal{CP}$, which generalises and strengthens the approach in~\cite{Igarashi:2009kj}, I have constructed the chiral symmetry currents, ward identities, and lattice chiral gauge actions and Majoranna fermions, and shown that a subset of the chiral operators considered here are local. I have also demonstrated that the inclusion of the Higgs in the lattice electroweak does not lead to any non localities if the correct chiral symmetry and $\mathcal{CP}$ symmetries are applied.

One troubling aspect, however, remains, namely the issue that the lattice charge conjugation matrix is non-local, and it is therefore not clear that this has a smooth limit to its continuum limit. This is certainly better than having non-local shifts in propagators and thus having observables without a clear continuum limit. If lattice QCD is to be understood in terms of a blocked continuum theory, this non-locality seems to be inevitable. If we abandon this understanding of lattice QCD, we abandon the theoretical basis of lattice chiral symmetry. My own belief is that the problem relates to the absence of doublers in the continuum, and can be understood, if not necessarily resolved, by blocking to the lattice theory from a continuum theory with an additional doubler field. I will investigate this in a subsequent work~\cite{Cundy:newpaper}. 

\section{Acknowledgements}
This research has been funded by the DFG grant FOR-465 and the SFB TR 55, and the BK21 program funded by NRF, Republic of Korea. I am grateful to Andreas Sch\"afer and Weonjong Lee for many useful discussions. 
\appendix
\section{$\mathcal{CP}$ in the continuum}\label{app:CP}
In the continuum, charge conjugation is defined as
\begin{align}
\psi_0(x)\rightarrow& -C^{-1}\overline{\psi}_0^T(x),&\overline{\psi}_0(x)\rightarrow &\psi_0(x)^TC,\nonumber\\
U(x,\mu)\rightarrow & U(x,\mu)^*,
\end{align}
where `$T$' denotes the transpose and `$*$' the complex conjugate, and the charge conjugation matrix $C$ satisfies
\begin{align}
C^{\dagger}C =& 1,&C^T =& -C,&C\gamma_{\mu}C^{-1} = &-\gamma_{\mu}^T,&C\gamma_5 C^{-1} = \gamma_5^T.
\end{align}
The Dirac operator $D_0$ transforms as
\begin{gather}
D_0[U](x,y) \rightarrow C^{-1}D_0[U^*](x,y)^T C,
\end{gather}

The Parity operation is defined as
\begin{align}
\psi_0(x) \rightarrow& \gamma_4 \psi_0(\overline{x}),&\overline{\psi}_0(x)\rightarrow& \overline{\psi}_0(\overline{x})\gamma_4,\nonumber\\
U(x,\mu)\rightarrow &U^P(x,\mu)=\left\{\begin{array}{l l}
U^{\dagger}(\overline{x} - a \hat{\mu},\mu)&\mu = 1,2,3\\
U(\overline{x},\mu)&\mu = 4
\end{array}\right.,
\end{align}
where
\begin{gather}
\overline{x} = (-x_1,-x_2,-x_3,x_4),
\end{gather}
and infinitesimal $a$.
In this case,
\begin{gather}
D_0[U](x,y)\rightarrow \gamma_4D_0[U^P](\overline{x},\overline{y})\gamma_4.
\end{gather}
The $\mathcal{CP}$ transformation in the continuum can be defined as
\begin{align}
\psi_0(x)\rightarrow&-W^{-1}\overline{\psi}_0^T(\overline{x}),&W^T=&W,&\overline{\psi}_0(x)\rightarrow&\overline{\psi}_0^T(\overline{x})W,\nonumber\\
U(x,\mu)\rightarrow&U^{CP}(x,\mu),
\end{align}
where
\begin{align}
W^{\dagger}W=&1&W\gamma_{\mu}W^{-1} =& \left\{\begin{array}{l l}
\gamma_{\mu}^{T}&\mu = 1,2,3\\
-\gamma_{\mu}^T&\mu = 4\end{array}\right.& W\gamma_5W^{-1} = -\gamma_5^T.\label{eq:gammacp}
\end{align}
Under this transformation,
\begin{gather}
 W (\mathcal{CP}:D_0[U](x,y))W^{-1} \rightarrow  D_0[U^{CP}](\overline{x},\overline{y})^T.\label{eq:121}
\end{gather}
The continuum massless action transforms under $\mathcal{CP}$ according to
\begin{align}
\overline{\psi}_0(x)D_0[U](x,y)\psi_0(y)\rightarrow& -\psi_0(\overline{x})^T W W^{-1} (D_0[U^{CP}](\overline{x},\overline{y}))^TW W^{-1} \overline{\psi}_0(\overline{y})^T\nonumber\\
=&\overline{\psi}_0(\overline{y})D_0[U](\overline{y},\overline{x})\psi_0(\overline{x}),
\end{align}
where there is a minus sign from fermion anti-commutation, and thus this action is invariant under $\mathcal{CP}$. Similarly, for the chiral decomposition of the action
\begin{align}
\overline{\psi}_0(x)D_0[U](x,y)\psi_0(y) = \frac{1}{4}\overline{\psi}_0(x)&(1+\gamma_5)D_0[U](x,y)(1-\gamma_5)\psi_0(y) + \nonumber\\&\frac{1}{4}\overline{\psi}_0(x)(1-\gamma_5)D_0[U](x,y)(1+\gamma_5)\psi_0(y), 
\end{align}
both of the Weyl fermion actions are invariant under CP.
\section{Properties of the $\gamma^{(\eta)}$ matrices}\label{app:glgr}
The equivalents of $\gamma_5$ in this formulation are $\gamma^{(\eta)}_L$ and $\gamma^{(\eta)}_R$, defined by equation (\ref{eq:gamma_eta1}). In the formulation given below, I have deflated the zero mode doublers to ensure that these expressions are well-defined.  
\begin{align}
{\gamma}^{(\eta)}_R = & \left[\gamma_5 \cos\left[\frac{1}{2}(1+\eta)(\pi - 2 \theta)\right] + \sign(\gamma_5 (D^{\dagger}-D)) \sin\left[ \frac{1}{2}(1+\eta)(\pi - 2 \theta)\right]\right](1-\ket{\phi_2}\bra{\phi_2})\nonumber\\& - \sign(\sin((\eta+\epsilon)\pi/2))\ket{\phi_2}\bra{\phi_2}\gamma_5\nonumber\\
{\gamma}^{(\eta)}_L = & \left[\gamma_5 \cos\left[\frac{1}{2}(\eta-1)(\pi - 2 \theta)\right] + \sign(\gamma_5(D^{\dagger}-D)) \sin\left[ \frac{1}{2}(\eta-1)(\pi - 2 \theta)\right]\right](1-\ket{\phi_2}\bra{\phi_2}) \nonumber\\&+ \sign(\sin((\eta+\epsilon)\pi/2))\ket{\phi_2}\bra{\phi_2}\gamma_5\nonumber\\
\tan\theta =& 2 \frac{\sqrt{1-D^{\dagger} D/4}}{\sqrt{D^{\dagger} D}},\label{eq:gamma_eta}
\end{align}
where $\epsilon$ is some infinitesimal real number.
At odd integer $\eta$, the dependence on the zero mode doublers cancels out, and these operators are local. For non integer $\eta$, the operators will in general be non-local. For even integer $\eta+\epsilon$, they will also be non-local. The ambiguity in $\sign(\sin((\eta+\epsilon)\pi/2))$ for these values of $\eta$ should first be resolved by consistently replacing this term in the definition by either the identity operator or minus the identity operator. 

$\gamma_L$ and $\gamma_R$ have the following properties:
%Put these in lemma form I think rather than like this

\begin{lemma} \label{point:1} $\gamma_L^{(\eta)} = \gamma_5 \gamma_R^{(-\eta)} \gamma_5$.\end{lemma}
\begin{proof}{This follows by inspection.\qed}\end{proof}
\begin{lemma}\label{point:2} $\gamma_R^{(1)} = \gamma_5(1-D)$.\end{lemma}
\begin{proof}{
\begin{align}
\tilde{\gamma}^{(1)}_R =& \gamma_5 \cos(\pi - 2 \theta) + \sign(\gamma_5 (D^{\dagger} - D)) \sin(\pi - 2 \theta)\nonumber\\
=&\gamma_5 \left( 1 - \frac{D^{\dagger}D}{2}\right) + \frac{1}{2}\gamma_5(D^{\dagger} - D)\nonumber\\
=& \gamma_5(1-D),
\end{align}
where I have used the familiar result $D+D^{\dagger} = D^{\dagger}D$. The result follows immediately.
\qed}\end{proof}
\begin{lemma}\label{point:3}$\gamma_R^{(-1)} = \gamma_5$.\end{lemma}
\begin{proof}{The result follows by Inspection.\qed}\end{proof}
\begin{lemma}\label{point:4} For odd integer $\eta$, $\gamma_R^{(\eta_1)} \gamma_R^{(\eta_2)} = \gamma_5 \gamma_R^{\eta_2 - \eta_1 - 1}$.\end{lemma}
\begin{proof}
 For the zero modes and non-zero modes,
\begin{align}
\gamma_R^{(\eta_1)} \gamma_R^{(\eta_2)} =& \cos(1+\eta_1)(\pi/2 - \theta)\cos(1+\eta_2)(\pi/2 - \theta) + \nonumber\\
&\phantom{indent}\sin(1+\eta_1)(\pi/2 - \theta)\sin(1+\eta_2)(\pi/2 - \theta) +\nonumber\\
&\phantom{indent} \gamma_5\sign(\gamma_5(D^{\dagger} - D))\nonumber\\
&\phantom{doubleindent}(\cos(1+\eta_1)(\pi/2 - \theta) \sin(1+\eta_2)(\pi/2 - \theta) - \nonumber\\
&\phantom{doubleindent}\sin(1+\eta_1)(\pi/2 - \theta) \cos(1+\eta_2)(\pi/2 - \theta)\nonumber\\
=&\cos((\eta_2 - \eta_1)(\pi/2 - \theta)) + \gamma_5\sign(\gamma_5(D^{\dagger} - D)) \sin((\eta_2 - \eta_1)(\pi/2 - \theta))\nonumber\\
=&\gamma_5 \gamma_R^{\eta_2 - \eta_1 - 1}.\label{eq:notabene}
\end{align}
For the zero mode doublers,
\begin{gather}
\gamma_R^{(\eta_1)} \gamma_R^{(\eta_2)}\ket{\phi_2} = \sign(\cos((\eta_1 - \eta_2)\pi/2) - \cos((\eta_1 + \eta_2)\pi/2))\ket{\phi_2}
\end{gather}
while 
\begin{gather}
\gamma_5 \gamma_R^{\eta_2 - \eta_1 - 1}\ket{\phi_2} =  \sign(\cos((\eta_1 - \eta_2)\pi/2))\label{eq:nb2}
\end{gather}
For odd integer $\eta$ (and all other $\eta$ where $\cos((\eta_1 - \eta_2)\pi/2) > \cos((\eta_1 + \eta_2)\pi/2)$) combining equations (\ref{eq:notabene}) and (\ref{eq:nb2}) gives the result
\qed
\end{proof}
\begin{lemma}\label{point:5} $\gamma_L^{(\eta_1)} \gamma_L^{(\eta_2)} = \gamma_5 \gamma_L^{\eta_2 - \eta_1 + 1}$.\end{lemma}
\begin{proof} Combining lemmas \ref{point:1} and \ref{point:4} gives the result immediately.\qed \end{proof}
\begin{lemma} $(\gamma_L^{(\eta)})^2 = (\gamma_R^{(\eta)})^2 = 1$ \end{lemma}
\begin{proof}
This follows from lemmas \ref{point:3}, \ref{point:4} and \ref{point:5}.\qed
\end{proof}
\begin{lemma} $\gamma_R^{(\eta_1)} \gamma_R^{(\eta_2)} = \gamma_R^{(2\eta_1 - \eta_2)}\gamma_R^{(\eta_1)}$ and $\gamma_L^{(\eta_1)} \gamma_L^{(\eta_2)} = \gamma_L^{(2\eta_1 - \eta_2)}\gamma_L^{(\eta_1)}$.\end{lemma}
\begin{proof}This follows immediately from lemmas \ref{point:4} and \ref{point:5}.\qed\end{proof}
\begin{lemma}\label{point:8} The Ginsparg-Wilson equation, $\gamma^{(\eta)}_L D + D \gamma_R^{(\eta)} = 0$, is satisfied
\end{lemma}
\begin{proof}
For the non-zero modes, using \ref{point:1} and \ref{point:2}, we have
\begin{align}
{\gamma}^{(\eta)}_L D + D \gamma_R^{(\eta)} =& \gamma_5 \left(\gamma_R^{(-\eta)} \gamma_R^{(-1)} - \gamma_R^{(-\eta)} \gamma_R^{(1)} + \gamma_R^{(-1)}\gamma_R^{(\eta)} - \gamma_R^{(1)} \gamma_R^{(\eta)}\right). 
\end{align}  
Using 
\begin{gather}
\gamma_R^{(\eta_1)} \gamma_R^{(\eta_2)} = \gamma_5 \gamma_R^{\eta_2 - \eta_1 - 1},\label{eq:gamma_multiply}
\end{gather}
it immediately follows that
\begin{gather}
{\gamma}^{(\eta)}_L D + D \gamma_R^{(\eta)} = \gamma_R^{(\eta - 2)} - \gamma_R^{(\eta)} + \gamma_R^{(\eta)} - \gamma_R^{(\eta - 2)} = 0.
\end{gather}
The zero modes $\phi_0$, are eigenvectors of both $D$ and $\gamma_5$, with $D\phi_0 = 0$. Therefore the Ginsparg-Wilson equation is satisfied for these operators. 

The zero mode doublers are eigenvectors of both $D$ and $\gamma_5$, with $\gamma_L \phi_2 = - \gamma_R\phi_2$. Therefore the Ginsparg Wilson equation is also satisfied in this case.\qed
\end{proof}
\begin{lemma}
$\gamma_5 \gamma_R^{(\eta)} D = D \gamma_5 \gamma_R^{(\eta)}$.
\end{lemma}
\begin{proof}
 Using  lemmas \ref{point:4}, \ref{point:5} and \ref{point:8}:
\begin{align}
\gamma_5 \gamma_R^{(\eta)} D = &\gamma_L^{(-\eta)} \gamma_L^{(1)}D\nonumber\\
=&D\gamma_R^{(-\eta)}\gamma_R^{(1)} \nonumber\\
=&  D \gamma_5 \gamma_R^{(\eta)}.\label{eq:resultneededinCP}
\end{align}
\qed
\end{proof} 
\begin{lemma}$\gamma_5 \gamma_R^{(-1 - \eta)} \gamma_R^{(-\eta)}\gamma_R^{(-1 - \eta)}\gamma_5 = \gamma_R^{(\eta)}$.\end{lemma}
\begin{proof}
Using lemma \ref{point:4},
\begin{align}
\gamma_5 \gamma_R^{(-1 - \eta)} \gamma_R^{(-\eta)}\gamma_R^{(-1 - \eta)}\gamma_5 = & \gamma_R^{(0)}\gamma_R^{(-1 - \eta)}\gamma_5\nonumber\\
=&\gamma_5 \gamma_R^{(-2-\eta)}\gamma_R^{(-1)} = \gamma_R^{(\eta)}.\label{eq:secondresultneededinCP}
\end{align} 
\qed
\end{proof}

\section{Eigenvalues of the overlap operator, $\gamma_L$ and $\gamma_R$}\label{app:eigenvalues}
The overlap operator is
\begin{gather}
D_2 =1+\gamma_5\sign(K),
\end{gather}
and the squared Hermitian overlap operator,
\begin{gather}
DD^{\dagger} = 2 + \gamma_5 \sign(K) + \sign(K)\gamma_5,
\end{gather}
commutes with $\gamma_5$. This means that the non-zero eigenvalues of $DD^{\dagger}$ are degenerate, and $D_2^{\dagger}D_2$ can be written in a chiral basis
\begin{gather}
DD^{\dagger} = \left(\begin{array}{l l}
\lambda^2&0\\
0&\lambda^2
\end{array}\right),
\end{gather}
where
\begin{gather}
\gamma_5 = \left(\begin{array}{l l}
1&0\\
0&-1
\end{array}\right)
\end{gather}
The degenerate eigenvectors of $DD^{\dagger}$ are $|g_+\rangle$ and $|g_-\rangle$, where $\gamma_5|g_{\pm}\rangle = \pm|g_{\pm}\rangle$. Thus 
\begin{align}
\langle g_+|DD^{\dagger} |g_+\rangle =& \lambda^2 = 2 + 2 \langle g_+|\sign(\gamma_5 D_W) |g_+\rangle\nonumber\\
\langle g_-|DD^{\dagger} |g_-\rangle =& \lambda^2 = 2 - 2 \langle g_-|\sign(\gamma_5D_W) |g_-\rangle
\end{align}  
Since the matrix sign function is Hermitian and given that $[\sign(K)]^2 = 1$, I can  write (excluding some, as yet undefined, contribution from the zero modes and their partners and a potential phase in the off diagonal terms which can be absorbed into the eigenvectors):
\begin{gather}
\sign (\gamma_5 D_W) =\left(\begin{array}{l l}
\frac{\lambda^2}{2}-1&\lambda\sqrt{1-\frac{\lambda^2}{4}}\\
\lambda\sqrt{1-\frac{\lambda^2}{4}}&1-\frac{\lambda^2}{2}
\end{array}
\right).
 \end{gather}
The lattice Dirac operator can be written in this basis constructed from the lattice eigenvectors as
\begin{gather}
D = 2\cos\theta \left(\begin{array}{l l} 
\cos{\theta}&\sin{\theta}\\
-\sin\theta&\cos\theta\end{array}\right),\label{eq:Cga}
\end{gather} 
where $\theta$ is defined in equation (\ref{eq:gamma_eta1}), and the continuum Dirac operator, though in a different basis, constructed from the continuum Dirac operator eigenvectors, as
\begin{gather}
D_0 = \lambda_0 \left(\begin{array}{l l} 
\cos(\pi/2)&\sin(\pi/2)\\
-\sin(\pi/2)&\cos(\pi/2)\end{array}\right).
\end{gather}
$\gamma_R$ is defined as
\begin{gather}
\gamma_R^{(\eta)} = D^{-(1+\eta)/2}Z D_0^{(1+\eta)/2} \gamma_5 D_0^{-(1+\eta)/2} Z^{\dagger} D^{(1+\eta)/2},
\end{gather}
where $Z$ projects the basis for the continuum Dirac operator onto the basis for the lattice operator. It immediately follows that (if we exclude for the moment the zero modes and their partners),
\begin{align}
\gamma_R^{(\eta)} =& \left(\begin{array}{l l} 
\cos((1+\eta)(\theta-\pi/2)/2)&-\sin((\theta-\pi/2)(1+\eta)/2)\\
\sin((\theta-\pi/2)(1+\eta)/2)&\cos((\theta-\pi/2)(1+\eta)/2)\end{array}\right) 
\left(\begin{array}{l l} 
1&0\\
0&-1
\end{array}\right) \nonumber\\&\phantom{lotsofspaceandmorespace}
\left(\begin{array}{l l} 
\cos((1+\eta)(\theta-\pi/2)/2)&\sin((\theta-\pi/2)(1+\eta)/2)\\
-\sin((\theta-\pi/2)(1+\eta)/2)&\cos((\theta-\pi/2)(1+\eta)/2)\end{array}\right)\nonumber\\
=&
 \left(\begin{array}{l l} 
\cos((1+\eta)(\theta-\pi/2))&\sin((\theta-\pi/2)(1+\eta))\\
\sin((\theta-\pi/2)(1+\eta))&-\cos((\theta-\pi/2)(1+\eta))
\end{array}\right)\nonumber\\
=&\gamma_5 \cos((1+\eta)(\theta-\pi/2)) + \frac{1}{2}\gamma_5(D-D^{\dagger})\frac{1}{\sqrt{D^{\dagger}D}\sqrt{1-D^{\dagger}D/4}}\sin((\theta-\pi/2)(1+\eta)),\label{eq:Cgb}
\end{align}
which reduces to the definition in equation (\ref{eq:8a}). The explicit form of $\gamma_L^{(\eta)}$ can be constructed in a similar way. It can easily be confirmed that this equation also applies to the zero modes and their doublers for odd integer $\eta$ by explicitly calculating $\gamma_R |\phi_0\rangle$ and $\gamma_R |\phi_2\rangle$.

The eigenvectors of the Hermitian Dirac operator $H = \gamma_5 + \sign(K)$ are then (up to some phase)
\begin{align}
H\ket{H_+} =& \lambda \ket{H_+}\nonumber\\
H\ket{H_-} =& -\lambda \ket{H_-}\nonumber\\
\left(\begin{array}{c}
\ket{H_+}\\ \ket{H_-}\end{array}\right) = &
\left(\begin{array}{c c}
-\cos(\theta/2)&-\sin(\theta/2)\\
-\sin(\theta/2)&\cos(\theta/2)\end{array}
\right)
\left(\begin{array}{c}
\ket{g_+}\\ \ket{g_-}\end{array}\right).
\end{align}
The eigenvectors of $\gamma_R^{(\eta)}$ are
\begin{align}
\left(\begin{array}{c}
\ket{\phi_+^{(\eta)}}\\ \ket{\phi_-^{(\eta)}}\end{array}\right) = &
\left(\begin{array}{c c}
-\cos((\eta + 1)(\pi/2-\theta)/2)&\sin((\eta + 1)(\pi/2-\theta)/2)\\
\sin((\eta + 1)(\pi/2-\theta)/2)&\cos((\eta + 1)(\pi/2-\theta)/2)
\end{array}
\right)
\left(\begin{array}{c}
\ket{g_+}\\ \ket{g_-}\end{array}\right)
\end{align}
Therefore,
\begin{align}
\left(\begin{array}{c}
\ket{\phi_+^{(\eta)}}\\ \ket{\phi_-^{(\eta)}}\end{array}\right) = &
\left(\begin{array}{c c}
\cos(\alpha^{(\eta)})&\sin(\alpha^{(\eta)})\\
-\sin(\alpha^{(\eta)})&\cos(\alpha^{(\eta)})
\end{array}
\right)
\left(\begin{array}{c}
\ket{H_+}\\ \ket{H_-}\end{array}\right),
\end{align}
where
\begin{gather}
\alpha^{(\eta)} =(\theta + (\eta + 1)(\pi/2 - \theta))/2.
\end{gather}
Under $\mathcal{CP}$, the non-zero eigenvalue equations transform as
\begin{align}
-W(\gamma_5 + \sign(K))^T W^{-1}\ket{H_+^{CP}} =& \lambda\ket{H_+^{CP}}\nonumber\\
-W(\gamma_5 + \sign(K))^T W^{-1}\ket{H_-^{CP}} =& -\lambda\ket{H_-^{CP}}.
\end{align}
Taking the transpose gives
\begin{align}
(\ket{H_+^{CP}})^TW^{-1}(\gamma_5 + \sign(K))  =& -\lambda(\ket{H_+^{CP}})^TW^{-1}\nonumber\\
(\ket{H_+^{CP}})W^{-1}(\gamma_5 + \sign(K)) =& \lambda(\ket{H_-^{CP}})^TW^{-1}.
\end{align}
and therefore
\begin{align}
(\ket{H_+^{CP}})^TW^{-1} =& \theta_+ \bra{H_-}\nonumber\\
(\ket{H_-^{CP}})^TW^{-1} =& \theta_- \bra{H_+},
\end{align}
where $\theta_{\pm}$ are pure phases. These can be found by constructing $\gamma_5$ operator in the basis of $\ket{H_+}$ and $\ket{H_-}$, which gives
\begin{gather}
\left(\begin{array}{c c}
\bra{H_+}\gamma_5\ket{H_+}&\bra{H_+}\gamma_5\ket{H_-}\\
\bra{H_-}\gamma_5\ket{H_+}&\bra{H_-}\gamma_5\ket{H_-}
\end{array}\right) = \left(\begin{array}{c c} 
-\lambda/2&-\sqrt{1-\lambda^2/4}\\
-\sqrt{1-\lambda^2/4}&\lambda/2\end{array}\right).
\end{gather} 
By considering the behaviour of $\gamma_5$ under $\mathcal{CP}$, it can be shown that $\theta_+\theta_+^* = \theta_-\theta_-^* = - \theta_+\theta_-^* = 1$, and therefore
\begin{align}
(\ket{H_+^{CP}})^TW^{-1} =& - \bra{H_-}\nonumber\\
(\ket{H_-^{CP}})^TW^{-1} =&  \bra{H_+}.
\end{align}
The transformations of the zero modes and their partners can be found in the same way,
\begin{align}
(\ket{\phi_0^{CP}})^TW^{-1} =&  \bra{\phi_0}\nonumber\\
(\ket{\phi_2^{CP}})^TW^{-1} =&  \bra{\phi_2}.
\end{align}
However, the zero modes and their partners have the opposite chirality
\begin{align}
\bra{\phi_0^{CP}}\gamma_5^{CP}\ket{\phi_0^{CP}} = - \bra{\phi_0} \gamma_5 \ket{\phi_0}
\end{align} 
Finally, $B^{(\eta)}$ and ${\overline{B}}$ are defined as
\begin{align}
B^{(\eta)} =& D_0^{-(1+\eta)/2}Z^{\dagger} D^{(1+\eta)/2}\nonumber\\
\overline{B}^{(\eta)}=& D_0^{(1-\eta)/2}Z D_0^{-(1-\eta)/2}.
\end{align}
Under $\mathcal{CP}$, these transform as
\begin{align}
\mathcal{CP}:B^{(\eta)} = W(D_0^{-(1+\eta)/2})^TZ^* (D^{(1+\eta)/2})^TW^{-1} =& W(\overline{B}^{(-\eta)})^TW^{-1}\nonumber\\
\mathcal{CP}:\overline{B}^{(\eta)} = W(D^{(1-\eta)/2})^TZ^{T} (D_0^{-(1-\eta)/2})^TW^{-1} =&W ({B}^{(-\eta)})^TW^{-1}\label{eq:CPfromB}
\end{align}
 I also obtain, 
\begin{align}
\gamma_R^{(\eta)} =& (B^{(\eta)})^{-1} \gamma_5 B^{(\eta)}\nonumber\\
\mathcal{CP}:\gamma_R^{(\eta)} =& -W^{-1} (\overline{B}^{(-\eta)} \gamma_5 (\overline{B}^{(-\eta)} )^{-1})^T W\nonumber\\
=& - W^{-1} (\gamma_L^{(-\eta)} )^T W,
\end{align}
and similarly,
\begin{gather}
\mathcal{CP}:\gamma_L^{(\eta)} = - W^{-1} (\gamma_R^{(-\eta)})^T W
\end{gather}

\bibliographystyle{elsarticle-num}
\bibliography{weyl}

\end{document}